\RequirePackage{etoolbox}
\csdef{input@path}{%
 {sty/}
 {img/}
}%
\csgdef{bibdir}{bib/}

\documentclass[ba]{imsart}
\pubyear{0000}
\volume{00}
\issue{0}
\doi{0000}
\firstpage{1}
\lastpage{1}

\usepackage{amsthm}
\usepackage{amsmath}
\usepackage{natbib}
\usepackage[colorlinks,citecolor=blue,urlcolor=blue,filecolor=blue,backref=page]{hyperref}
\usepackage{graphicx}

\startlocaldefs

\usepackage{enumerate}
\usepackage{url} 
\usepackage{amssymb}
\usepackage{color}
\usepackage{bm}
\usepackage{lineno}
\usepackage{cleveref}
\usepackage{multirow}
\usepackage{sectsty}
\usepackage{array}
\usepackage{booktabs}
\usepackage{algorithm}
\usepackage{algorithmic}

\newtheorem{theorem}{Theorem}

\newtheorem{proposition}[theorem]{Proposition}
 
\theoremstyle{plain}

\newcommand{\bA}{\bm{A}}

\newcommand{\bC}{\bm{C}}

\newcommand{\bD}{\bm{D}}

\newcommand{\bG}{\bm{G}}
\newcommand{\I}{\bm{I}}
\newcommand{\bI}{\bm{I}}
\newcommand{\bH}{\bm{H}}

\newcommand{\bL}{\bm{L}}

\newcommand{\bR}{\bm{R}}

\newcommand{\bS}{\bm{S}}
\newcommand{\bT}{\bm{T}}

\newcommand{\bU}{\bm{U}}
\newcommand{\bu}{\bm{u}}

\newcommand{\X}{\bm{X}}
\newcommand{\bX}{\bm{X}}

\newcommand{\y}{\bm{y}}
\newcommand{\by}{\bm{y}}

\newcommand{\bz}{\bm{z}}

\newcommand{\bPsi}{\bm{\Psi}}

\newcommand{\bbeta}{\bm{\beta}}
\newcommand{\bepsilon}{\bm{\varepsilon}}

\newcommand{\bgamma}{\bm{\gamma}}
\newcommand{\bSigma}{\bm{\Sigma}}
\newcommand{\btheta}{\bm{\theta}}

\newcommand{\bOmega}{\bm{\Omega}}

\newcommand{\bzero}{\bm{0}}
\newcommand{\Real}{\mathrm{Re}}
\newcommand{\Imag}{\mathrm{Im}}
\newcommand{\hbb}{\hat{\bbeta}}

\DeclareMathOperator*{\Var}{Var}

\DeclareMathOperator*{\define}{=}

\newcommand{\MVN}{\mathrm{MVN}}

\newcommand{\BF}{\mathrm{BF} }

\newcommand{\UN}{\mathrm{Unif}}

\newcommand{\PPA}{\mathrm{PPA}}
\newcommand{\PIP}{\mathrm{PIP}}

\newcommand{\N}{M} 
\renewcommand{\P}{P}

\endlocaldefs

\begin{document}

\begin{frontmatter}
\title{Fast model-fitting of Bayesian variable selection regression using the iterative complex factorization algorithm}

\runtitle{The iterative complex factorization algorithm for BVSR} 

\begin{aug}
\author{\fnms{Quan} \snm{Zhou} \thanksref{addr1,note1}
\ead[label=e1]{qz9@rice.edu}
}
\and
\author{\fnms{Yongtao} \snm{Guan} \thanksref{addr2,fund1} 
\ead[label=e2]{yongtaog@bcm.edu}
}

\runauthor{Q. Zhou and Y. Guan}

\address[addr1]{
Department of Statistics, Rice University, 6100 Main St, Houston TX, 77005. 
Email: \printead{e1}
}

\address[addr2]{
USDA/ARS Children's Nutrition Research Center, Department of Pediatrics, and Department of Molecular and Human Genetics of Baylor College of Medicine, 1100 Bates, Room 2070, Houston TX, 77030. 
Email: \printead{e2}
}

\thankstext{note1}{Quan Zhou did most of his work when he was a PhD student in the program of Quantitative and Computational Biosciences, Baylor College of Medicine.}

\thankstext{fund1}{
This work was supported by United States Department of Agriculture/Agriculture Research Service under contract number 6250-51000-057 and National Institutes of Health under award number R01HG008157. 
}
\end{aug}

\begin{abstract}
Bayesian variable selection regression (BVSR) is able to jointly analyze genome-wide genetic datasets, but the slow computation via Markov chain Monte Carlo (MCMC) hampered its wide-spread usage. 
Here we present a novel iterative method to solve a special class of linear systems, which can increase the speed of the BVSR model-fitting tenfold. 
The iterative method hinges on the complex factorization of the sum of two matrices and the solution path resides in the complex domain (instead of the real domain). 
Compared to the Gauss-Seidel method, the complex factorization converges almost instantaneously and its error is several magnitude smaller than that of the Gauss-Seidel method. 
More importantly, the error is always within the pre-specified precision while the Gauss-Seidel method is not. 
For large problems with thousands of covariates, the complex factorization is 10 -- 100 times faster than either the Gauss-Seidel method or the direct method via the Cholesky decomposition. 
In BVSR, one needs to repetitively solve large penalized regression systems whose design matrices only change slightly between adjacent MCMC steps. 
This slight change in design matrix enables the adaptation of the iterative complex factorization method. 
The computational innovation will facilitate the wide-spread use of BVSR in reanalyzing genome-wide association datasets.  
\end{abstract}

\begin{keyword}
\kwd{Cholesky decomposition} 
\kwd{exchange algorithm}
\kwd{fastBVSR}
\kwd{Gauss-Seidel method}
\kwd{heritability}
\end{keyword}

\end{frontmatter}

 
\section{Introduction}
Bayesian variable selection regression (BVSR) can jointly analyze genome-wide genetic data to produce the posterior probability of association for each covariate and estimate hyperparameters such as heritability and the number of covariates that have nonzero effects~\citep{guan2011bayesian}.  
But the slow computation due to model averaging using Markov chain Monte Carlo (MCMC) hampered its otherwise warranted wide-spread usage.  
Here we present a novel iterative method to solve a special class of linear systems, which can increase the speed of the BVSR model-fitting tenfold. 
 
\subsection{Model and priors}  
We first briefly introduce the BVSR method. 
Our model, prior specification and notation follow closely those of~\citet{guan2011bayesian}.    
Consider the linear regression model 
\begin{equation}
\by = \mu \bm{1} + \bX \bbeta + \bepsilon, \quad \bepsilon \sim \MVN (\bzero, \tau^{-1} \bI ), 
\end{equation}
where $\bX$ is an $n \times \N$ column-centered matrix with $\N \gg n$, $\bI$ denotes an identity matrix of proper dimension, 
$\by$ and $\bepsilon$ are $n$-vectors, $\bbeta$ is an $\N$-vector, and MVN stands for multivariate normal distribution.
Let $\gamma_j$ be an indicator of the $j$-th covariate having a nonzero effect and write $\bgamma = \{\gamma_1, \dots, \gamma_j, \dots,  \gamma_\N \}$. A spike-and-slab prior for $\beta_j$ (the $j$-th component of $\bbeta$) is specified below,
\begin{equation}\label{eq:bvsr}
\begin{aligned}
\gamma_j \mid \pi  & \sim \text{Bernoulli}(\pi), \\
\beta_j  \mid \gamma_j = 0 & \sim \delta_0 , \\
\beta_j \mid \gamma_j = 1,  \sigma_\beta,  \tau  & \sim N (0, \sigma^2_\beta / \tau ), \\
\end{aligned}
\end{equation}
where $\pi$ is the proportion of covariates that have non-zero effects 
and $\sigma_\beta^2$ is the variance of prior effect size (scaled by $\tau^{-1}$). We will specify priors for both later. 
We use noninformative priors on the parameters $\mu$ and $\tau$,
\begin{equation}\label{eq:prior1}
\begin{aligned}
\mu \mid \tau  & \sim  N(0, \sigma_\mu^2 / \tau) ,     & \sigma_\mu  \rightarrow \infty , \\
\tau & \sim   \text{Gamma}(\kappa_1 / 2, \kappa_2 / 2)  ,  \;  & \kappa_1, \kappa_2  \rightarrow 0 ,  \\
\end{aligned}
\end{equation}
where $\text{Gamma}$ is in the shape-rate parameterization. 
As pointed out in~\citet{guan2011bayesian}, prior~\eqref{eq:prior1} is equivalent to $\P(\mu, \tau) \propto \tau^{-1/2}$, which is known as Jeffreys' prior~\citep{ibrahim1991bayesian,o2004kendall}. 
In practice, this means to use a diffuse prior for both $\mu$ and $\tau$.
Some may favor a simpler form $\P(\mu, \tau) \propto 1/\tau$, which makes no practical difference~\citep{berger2001objective, liang2008mixtures}.

Given $\bgamma$ and $\sigma_\beta^2$, after integrating out $\bbeta$, $\tau$, $\mu$ and letting $\sigma_\mu  \rightarrow \infty$, $\kappa_1, \kappa_2  \rightarrow 0 $, the Bayes factor with reference to the null model can be computed in closed form,    
\begin{equation}\label{eq:bf}
\BF (\bgamma, \sigma_\beta^2) =   | \bI + \sigma^2_\beta \bX_{\bgamma}^t  \bX_{\bgamma} |^{-1/2} \left(  1 - \dfrac{ \by^t \bX_{\bgamma} \hat{\bbeta} }
{ \by^t \by -  n \bar{y}^2 } \right)^{-n/2},
\end{equation}
where $\bX_{\bgamma}$ denotes the submatrix of $\bX$ with columns for which $\gamma_j = 1$,  $| \cdot |$ denotes matrix determinant, and $\hat{\bbeta}$ is the posterior mean for $\bbeta$ given by
\begin{equation}\label{eq:hat.beta}
\hat{\bbeta} \define  ( \bX_{\bgamma}^t  \bX_{\bgamma}  + \sigma^{-2}_\beta \bI )^{-1}  \bX_{\bgamma}^t  \by .
\end{equation}
The null-based Bayes factor $\BF (\bgamma, \sigma_\beta^2)$ is proportional to the marginal likelihood $P (\by  \mid \bgamma, \sigma_\beta^2)$, and evaluating $\BF$ is easier than evaluating the marginal likelihood due to cancellation of constants. 
The limiting prior~\eqref{eq:prior1} not only makes the Bayes factor expression simpler (compared to that with finite value of $\sigma_\mu$ and positive values of $\kappa_1$ and $\kappa_2$), but also makes it invariant with respect to the shifting and scaling of $\by$.

We now discuss the prior specification for the two hyperparameters, $\pi$ and $\sigma_\beta^2$.
To specify the prior for $\sigma_\beta^2$,  \citet{guan2011bayesian} introduced a hyperparameter $h$ (which stands for heritability) such that
\begin{equation}\label{eq:def.h}
h   =  \sigma^2_\beta \sum\limits_{j=1}^\N  \gamma_j s_j  /  ( 1 +  \sigma^2_\beta \sum\limits_{j=1}^\N  \gamma_j s_j   ) , 
\end{equation}
where $s_j$ denotes the variance of the $j$-th covariate.  
Conditional on $\bgamma$, specifying a prior on $h$ will induce a prior on $\sigma_\beta^2$, so henceforth we may write $\sigma_\beta^2(\bgamma, h)$ to emphasize $\sigma_\beta^2$ is a function of $h$ and $\bgamma$. 
Since $h$ is motivated by the narrow-sense heritability, its prior is easy to specify and we use   
 \begin{equation}\label{eq:h}
h \sim \UN(0, 1),
\end{equation}
by default to reflect our lack of knowledge of heritability, although one can impose a strong prior by specifying a uniform distribution on a narrow support.  
A bonus of specifying the prior on $\sigma_\beta^2$ through $h$ is that when $h  \sim \UN(0, 1)$, the induced prior on $\sigma_\beta^2$ is heavy-tailed~\citep{guan2011bayesian}. 

Up till now, we follow faithfully the model and the prior specification of \citet{guan2011bayesian}.  
The prior on $\bgamma$ can be induced by the prior on $\pi$. 
\citet{guan2011bayesian} specified the prior on $\pi$ as uniform on its log scale,  
$\log \pi  \sim \UN(\log\pi_{min}, \log\pi_{max})$,  
which is equivalent to $\P(\pi) \propto 1/\pi$ for $\pi \in (\pi_{min}, \pi_{max})$, and sampled $\pi, h, \bgamma$. 
But here we do something slightly different by integrating out $\pi$ analytically. This is sensible because $\bgamma$ is very informative on $\pi$. Specifically,  
we integrate $\P(\bgamma, \pi)$ over  $\P(\pi)$ such that $\P (\bgamma) = \int \P (\bgamma \mid \pi) \P (\pi) d\pi$  to obtain the marginal prior on $\bgamma$,
\begin{equation}\label{eq:prior.gamma}
\begin{aligned}
\P (\bgamma) 
= \frac{1}{\log{(\pi_{max}/\pi_{min}})} \int_{\pi_{min}}^{\pi_{max}} \P(\bgamma \mid \pi) /\pi\; d \pi,  
\end{aligned}
\end{equation}
where the finite integral is related to the truncated Beta distribution and can be evaluated conveniently. 
If $\pi_{min}$ goes to $0$ and $\pi_{max}$ goes to $1$ we have an improper prior $\P(\pi) \propto 1/\pi$ and the marginal prior on $\bgamma$ becomes
\begin{equation}\label{eq:prior.gamma2}
\begin{aligned}
\P(\bgamma) \propto \;  \Gamma( |\bgamma| ) \, \Gamma (\N + 1 - |\bgamma|), 
\end{aligned}
\end{equation}
where we recall $\N$ is the total number of covariates, $|\bgamma|= \sum{\gamma_j}$ is the number of selected covariates in the model, and $\Gamma$ denotes the Gamma function. 
$\P(\bgamma)$ is always a proper probability distribution because  it is defined on a finite set.

\subsection{Posterior inference and computation}  
The joint posterior distribution of $(\bgamma, h)$ is given by 
\begin{equation}\label{eq:post}
\begin{aligned}
\P(\bgamma, h \mid \by) \propto P (\by \mid \bgamma, h) \P(\bgamma) \P(h). 
\end{aligned}
\end{equation}
The posterior inferences typically include computing the posterior inclusion probability $P(\gamma_j = 1 \mid \by)$, which measures the strength of marginal association of the $j$-th covariate, the posterior distribution of the model size $|\bgamma|$, and the posterior distribution of the heritability $h$, which measures the proportion of phenotypic variance explained by the selected models.
We use MCMC to sample this joint posterior of $(\bgamma, h)$. Our sampling scheme follows closely that of~\citet{guan2011bayesian}. 
In each MCMC iteration, to evaluate~\eqref{eq:post} for a proposed parameter pair $(\bgamma', h')$, we need to compute the marginal likelihood $P (\by \mid \bgamma', h')$, which is proportional to~\eqref{eq:bf}. Two time-consuming calculations are the matrix determinant $|\bI + \sigma^2_\beta \bX_{\bgamma}^t  \bX_{\bgamma} |$ and $\hat{\bbeta}$ defined in~\eqref{eq:hat.beta}, both of which have cubic complexity (in $|\bgamma|$). 
The computation of the determinant can be avoided by using a MCMC sampling trick which we will discuss later. 
The main focus of the paper is a novel algorithm to evaluate~\eqref{eq:hat.beta}, which reduces its complexity from cubic to  quadratic. 

The rest of the paper is structured as follows. 
Section 2 introduces the iterative complex factorization (ICF) algorithm. 
Section 3 describes how to incorporate the ICF algorithm into BVSR. Both sections contain numerical examples, including a real dataset from genome-wide association studies. 
A short discussion concludes the paper.

\section{The iterative complex factorization}
In this section, we propose a novel algorithm for solving the following linear system 
\begin{equation}\label{eq:beta2}
(\bX^t \bX + \bSigma^2) \hat{\bbeta}  = \bz , 
\end{equation}
where $\bSigma$ is a diagonal matrix with positive (but not necessarily identical) entries on the diagonal, and $\bX$ is an $n\times p$ matrix.  
Clearly \eqref{eq:hat.beta} is a special case of \eqref{eq:beta2}. 
In the context of BVSR, $\bX$ should be understood as $\bX_{\bgamma}$ and $p=|\bgamma|$. 
We assume $p < n$, where the sample size $n$ ranges from several hundreds to tens of thousands.
The computational advancement we will introduce, however, can be applied to scenarios where $p > n$ (see discussion in Section~\ref{sec:disc}).

Note that $\hat{\bbeta}$ is the familiar ridge regression estimator~\citep{draper1979ridge}. 
It may appear that an algorithm designed for solving ridge regression can be borrowed to solve~\eqref{eq:beta2}. 
But a unique feature of BVSR is that in each iteration of MCMC, the design matrix $\bX$ usually changes only by one or a few columns. 
Thus, $\bX^t \bX$ and its Cholesky decomposition can be obtained conveniently 
(details will follow).  
This unique feature allows us to design a much more efficient algorithm.

\subsection{Existing methods}
In \cite{guan2011bayesian}, the linear system~\eqref{eq:beta2} was solved using the Cholesky decomposition of $\bX^t \bX  +  \bSigma^2$,  which requires $ p^3/3$ flops~\citep[Lec.~23]{trefethen1997numerical}. 
Although computing $\bX^t \bX$ from $\bX$ requires $O(np^2)$ flops, 
when $\bX$ only changes by a few columns, the majority of the entries in $\bX^t \bX$ do not change and updating $\bX^t \bX$ only requires $O(np)$ flops.

Iterative methods sometimes can be used to reduce the computational time. 
Define
\begin{equation}\label{def:a}
\bA \define \bX^t \bX  +  \bSigma^2 = \bL + \bD + \bU
\end{equation}
where $\bL$ ($\bU$) is the strictly lower (upper) triangular component and $\bD$ contains only the diagonals. 
Then three popular iterative procedures can be summarized as follows: 
\begin{equation*}\label{def:3iter}
\begin{aligned}
\begin{array}{cc}
\text{Jacobi method:} &   \hat{\bbeta}^{(k+1)} = \bD^{-1} \left[ -(\bL + \bU) \hat{\bbeta}^{(k)} + \bz  \right ] ;  \\
\text{Gauss-Seidel method:} &   \hat{\bbeta}^{(k+1)} =( \bD + \bL)^{-1} \left( -\bU \hat{\bbeta}^{(k)} + \bz  \right ) ;  \\
\text{successive over-relaxation:} &   \hat{\bbeta}^{(k+1)} = (\bD + \omega \bL)^{-1} \left[ -(\omega \bU - (1-\omega)\bD ) \hat{\bbeta}^{(k)} + \omega \bz  \right ]   .
\end{array}
\end{aligned}
\end{equation*}
The successive over-relaxation (SOR) method is a generalization of the Gauss-Seidel method, where $\omega$ is called the relaxation parameter. 
When $\bA$ is positive definite, the Gauss-Seidel method always converges and the SOR method converges  for $\omega \in (0, 2)$~\citep[Chap.~10.1.2]{golub2012matrix}.
For all three iterative methods, each iteration requires $2p^2$ flops. Thus, whether an iterative method is more efficient than the Cholesky decomposition depends on how many iterations it takes to converge.  
Another notable class of iterative methods is called Krylov subspace methods.  Two famous examples are the steepest descent and the conjugate gradient~\citep[Lec.~38]{trefethen1997numerical}.  

In principle, all methods developed to solve ridge regression can be used here to solve~\eqref{eq:beta2}, as we alluded to earlier. 
For example, the methods of \cite{elden1977algorithms}, \cite{lawson1995solving} and~\cite{turlach2006even} were developed for solving~\eqref{eq:beta2} with fixed $\bX$ but changing $\bSigma$, 
while  \cite{hawkins2002faster} devised a method  for fixed $\bSigma$ but changing $\bX$.  
Modern least square solvers~\citep{rokhlin2008fast, avron2010blendenpik, meng2014lsrn} typically considered the case where $\bX$ is extremely large, sparse, or ill-conditioned, and under such conditions, these least square solvers outperform solving directly via Cholesky decomposition; 
see~\citet[Sec. 5]{meng2014lsrn} for how to apply least square solvers to ridge regression.
But all methods quoted above are less effective (or not applicable) in the context of BVSR, because they are not designed for solving~\eqref{eq:beta2} millions of times, each time with a slightly different $\bX$ and a different $\bSigma$, and they did not take advantage of the feature of BVSR that the Cholesky decomposition of $\bX^t\bX$ can be obtained efficiently.

\subsection{The iterative complex factorization (ICF) algorithm}
Let $\bR^t\bR$ be the Cholesky decomposition of $\bX^t\bX$, where $\bR$ is upper triangular.   
In the context of BVSR, given the Cholesky decomposition of $\bX^t \bX$, the Cholesky decomposition of a new matrix $(\bX')^t \bX'$ can be obtained efficiently since $\bX'$ (the proposed new design matrix) differs from $\bX$ only by one or a few columns (see Section~\ref{sec:bvsr.icf}).  
So we consider solving the linear system
\begin{equation}\label{eq:ridge3}
(\bR^t \bR + \bSigma^2) \hat{\bbeta} = \bz. 
\end{equation}
Contrary to our intuition, the Cholesky decomposition of $\bR^t \bR + \bSigma^2$ cannot be obtained efficiently.  This was also noticed in~\citet[p. 6]{zhou2013polygenic}.
We instead perform the following decomposition 
\begin{equation}\label{eq:icf.defs}
\begin{aligned}
\bR^t \bR + \bSigma^2 & = \bH - i \bS \\
\bH \,  &\define \, (\bR^t - i \bSigma )(\bR + i \bSigma ) \\
\bS  \, &\define \, \bR^t \bSigma - \bSigma \bR,
\end{aligned}
\end{equation}
where $i$ is the imaginary unit. Then we have the update 
$$\bH \hbb^{(k+1)} =  i\bS\hbb^{(k)} + \bz,$$ 
where the right-hand side is a complex vector, and  $\hbb^{(k+1)}$ can be obtained by a forward and a backward substitution involving two complex triangular matrices,  $\bR^t - i \bSigma$ and $\bR + i \bSigma$. 
This update, however, diverges from time to time. Examining the details of the observed divergent cases reveals that the culprit is the imaginary part of $\hbb^{(k)}$. 
Because the solution $\hbb$ is real, discarding the imaginary part of $\hbb^{(k)}$ at the end of each iteration will not affect the fixed point to which the iterative method converges. 
Denoting the real part of a complex entity (scalar or vector) by $\Real$,  the generalized update of our algorithm ICF (Iterative Complex Factorization) becomes 
\begin{equation}\label{eq:icf.iter2}
\begin{aligned}
\hat{\bbeta}^{(k+1)} &=  \Real [     (1-\omega)  \hat{\bbeta}^{(k)}   +   \omega \bH^{-1}   ( i \bS    \hat{\bbeta}^{(k)}   +  \bz  ) ], 
\end{aligned}
\end{equation}
where we have also introduced a relaxation parameter $\omega$.  Intuitively $\omega$ makes the update \emph{lazy} to avoid over-shooting. 
The revised update converges almost instantaneously, in a few iterations, compared to a few dozen to a few hundred iterations with the Gauss-Seidel method.  
Each iteration of~\eqref{eq:icf.iter2} requires $6p^2$ flops,  thrice that required by a Gauss-Seidel iteration, because ICF operates complex (instead of real) matrices and vectors.  
The right-hand side of~\eqref{eq:icf.iter2} can be reorganized as  $ \Real [    ((1-\omega) \bI  +  i \omega \bH^{-1}   \bS ) \hbb^{(k)} + \omega \bH^{-1} \bz ].$ Note that $\bH^{-1} \bz$ can be computed via forward and backward substitutions and does not require matrix inversion.   

\subsection{ICF converges to the right target}\label{sec:icf.conv}
\begin{proposition}\label{prop1}
Denote $\bPsi(\omega) =  \Real [    (1-\omega) \bI  +  i \omega \bH^{-1}   \bS   ].$  
Then ICF in~\eqref{eq:icf.iter2} converges for any starting point if  $\rho(\bPsi(\omega)) < 1$ where $\rho$ denotes the spectral radius.
\end{proposition} 
\begin{proof}
The true solution $\hat{\bbeta}$ satisfies $  \bH \hbb =  i\bS\hbb  + \bz$, which, after some algebra, gives
$\hat{\bbeta}^{(k + 1)} - \hat{\bbeta} = \bPsi(\omega) (\hat{\bbeta}^{(k)} - \hat{\bbeta} )$.  
The statement then follows faithfully from Theorem~10.1.1 of~\citet{golub2012matrix}. 
\end{proof}

\begin{theorem}\label{th:icf}
There exists $\omega \in (0, 1]$ such that the ICF update detailed in~\eqref{eq:icf.iter2} converges to the true solution.  
\end{theorem}
\begin{proof}
Using the notation defined in~\eqref{def:a} and~\eqref{eq:icf.defs}, we have  
 $\bH = \bA + i \bS$.  
Since $[\Real(\bH^{-1}) + i \, \Imag(\bH^{-1}) ](\bA + i \bS) = \bI $, 
 we have 
\begin{align*}
 \left\{\begin{array}{c}
 \Real(\bH^{-1}) \bA -  \Imag(\bH^{-1})\bS = \bI, \\
  \Imag(\bH^{-1}) \bA + \Real(\bH^{-1}) \bS = \bm{0}. 
\end{array}  
 \right.  
\end{align*}  
Using the fact that $\bA$ is invertible, we can solve the above system and obtain 
\begin{equation*}
\Imag(\bH^{-1}) = - \bA^{-1} \bS (\bA + \bS \bA^{-1} \bS )^{-1}  .
\end{equation*}
Both $\bA$ and $\bS$ are real matrices, and by the Woodbury identity we have 
\begin{equation}\label{eq:psi2}
\bPsi(\omega) = \bI - \omega (\bI + (\bA^{-1} \bS)^2 )^{-1}.
\end{equation} 
Immediately, for a fixed $\omega$, the spectrum of $\bPsi(\omega)$ is fully determined by the spectrum of $\bA^{-1} \bS$. 
Because $\bS$ is skew-symmetric, we have $-\bS = \bS^t$.  Since $\bA$ is symmetric, so does $\bA^{-1/2}.$ 
Then we have $-\bA^{-1/2} \bS \bA^{-1/2} = \bA^{-1/2}\bS^t \bA^{-1/2}$, which means $\bA^{-1/2} \bS \bA^{-1/2}$ is also skew-symmetric.  
Hence the eigenvalues of  $\bA^{-1} \bS$, which are identical to those of $\bA^{-1/2} \bS \bA^{-1/2}$, are conjugate pairs of pure imaginary numbers or zero. 
Let $\pm  \eta i $ be such a pair with $\eta \geq 0 $ and $\bu$ be the eigenvector corresponding to the eigenvalue $\eta i$. 
We have $\bA^{-1} \bS \bu = i \eta \bu$, which implies $i \bu^* \bS \bu = - \eta \bu^* \bA \bu$,
where $\bu^*$ denotes the conjugate transpose of $\bu$. 
Since $\bH$ is a Hermitian positive definite matrix, 
\begin{equation*}
\bu^* \bH \bu =  \bu^* (\bA + i\bS ) \bu = (1 - \eta) \bu^* \bA \bu > 0  .
\end{equation*}
Since $\bA$ is also positive definite,  we have $\bu^* \bA \bu > 0$ and 
\begin{equation} \label{eq:eta}
0 \le \eta < 1. 
\end{equation}
By~\eqref{eq:psi2},  the eigenvalues of $\bPsi(\omega)$ are identical pairs equal to $1- \omega/(1-\eta^2)$ (this includes the case $\eta=0$).  
Proposition~\ref{prop1} just requires $\left | 1- \omega/ (1-\eta^2)  \right | < 1$, or equivalently,
\begin{equation}\label{eq:cond}
 0 < \omega < 2(1 - \eta^2),
\end{equation}
holds for all $\eta$. Thus the existence of $\omega$ follows from \eqref{eq:eta}.
\end{proof}


By Proposition~\ref{prop1}, the spectral radius of $\bPsi(\omega)$ determines how fast the error converges to zero. 
We provide a theory-guided procedure to adaptively tune the relaxation parameter $\omega$, which relies on 
the following proposition that connects  $\rho(\bPsi(\omega) )$ with $\omega$ and $\eta$. 
\begin{proposition}\label{prop:adapt}
Denote the eigenvalues of $\bA^{-1}\bS$ as $\pm \eta i \; (\eta \geq 0)$, where $\bA$ and $\bS$ are defined in ~\eqref{def:a} and~\eqref{eq:icf.defs}, and obtain $\eta_{min}$ and $\eta_{max}$. 
Then the spectral radius of $\bPsi(\omega)$ is 
\begin{equation}\label{eq:rad.psi}
 \rho (\bPsi(\omega) ) =  \max \left\{ 1 - \dfrac{\omega}{1 - \eta_{min}^2} \; , \quad 
 \dfrac{\omega}{1 - \eta_{max}^2} - 1  \right\}, 
\end{equation}
and the optimal value for $\omega$ to achieve the minimum of  $\rho (\bPsi(\omega) )$ is 
\begin{equation}\label{eq:optimal.w}
\omega^\star = 2 \left( \dfrac{1 }{1 - \eta_{min}^2 } + \dfrac{1}{ 1 - \eta_{max}^2 } \right)^{-1}. 
\end{equation}
\end{proposition}   
\begin{proof}
By~\eqref{eq:psi2} and~\eqref{eq:eta}, the smallest and the largest eigenvalue of $\bPsi(\omega)$ are $1 - \omega / (1 - \eta^2_{max})$ and $1 - \omega / (1 - \eta^2_{min}) $ with $\eta \in [0, 1)$. After adjusting for their signs, we obtain~\eqref{eq:rad.psi}. 
The right-hand side of~\eqref{eq:rad.psi} is a function of $\omega$, with the first item decreasing linearly in $\omega$ and the second item increasing linearly in $\omega$. 
Hence the minimum of ~\eqref{eq:rad.psi} is attained when the two quantities in the braces are equal, which proves~\eqref{eq:optimal.w}. 
\end{proof}
Our adaptive strategy for choosing $\omega$ assumes that $\eta_{min}$ is zero, which  holds trivially for odd $p$ by the property of skew-symmetric matrices. 
When $p$ is even, our numerical studies found $\eta_{min} = 0$ is still a valid assumption in practice (Supplementary  S3).
We start the ICF update with $\omega^{(0)} = 1$. 
Suppose at the $k$-th iteration we can produce an estimate $\hat{\rho}^{(k)}$ for the spectral radius of $\bPsi ( \omega^{(k)} )$. Then using~\eqref{eq:rad.psi}, $\eta_{max}^2$ can be estimated by
${\eta_{max}^2} \approx 1- \omega^{(k)}/ (1+\hat{\rho}^{(k)}).$
Plugging this into~\eqref{eq:optimal.w} we obtain an update for $\omega$ 
\begin{equation*}
\omega^{(k+1)} = \dfrac{2 \omega^{(k)}}{1 +  \omega^{(k)} + \hat{\rho}^{(k)}  }.
\end{equation*}
Note that the update does not involve $\eta_{max}^2$, but only $\omega^{(k)}$ and $\hat{\rho}^{(k)}$, 
and $\omega^{(k+1)}$ is a decreasing function of $\rho^{(k)}$. 
Finally, to estimate the spectral radius of $\rho^{(k)}$ we use
\begin{equation*}
\hat{\rho}^{(k)} = \dfrac{ || \hat{\bbeta}^{(k)} - \hat{\bbeta}^{(k-1)} ||_2 }{|| \hat{\bbeta}^{(k-1)} - \hat{\bbeta}^{(k-2)} ||_2   }, 
\end{equation*}
where $|| \cdot ||_2$ denotes the $\ell^2$-norm. 
This update strategy borrows the idea of power iteration and is motivated by the observation that $\hat{\bbeta}^{(k)} - \hat{\bbeta}^{(k-1)} = \bPsi(\omega) (\hat{\bbeta}^{(k-1)} - \hat{\bbeta}^{(k-2)} )$ if $\omega$ were fixed (see also Proposition~\ref{prop1}). 
The procedure works well in our numerical studies. 
To take care of the boundary conditions, we use $\omega^{(k)} = 1$ for $k=0,1,2$. 

\subsection{ICF outperforms other methods}\label{sec:sim.icf}
Our numerical comparison studies were based on real datasets of genome-wide association studies downloaded from dbGaP.  
The details of the datasets can be found in Section~\ref{sec:iop}. 
Because the convergence of iterative methods is sensitive to the collinearity in the design matrix $\bX$,  our comparison studies used two datasets: the first one contains $20$K SNPs sampled across the whole genome, and the other contains $20$K SNPs that are physically adjacent.  
The first dataset has little or no collinearity (henceforth referred to as IND), and the second has collinearity due to linkage disequilibrium (henceforth referred to as DEP). 
The sample size is $n = 3000$ for both datasets. 
Our numerical studies compared different methods (detailed below) for their speed and accuracy of solving~\eqref{eq:beta2}. 
Given $p$, for one experiment we sampled without replacement $p$ columns from the IND (or DEP) dataset to obtain $\bX$,  simulated under the null $\bz \sim \MVN(\bzero, \bI)$, and solved~\eqref{eq:beta2} using different methods with $\bSigma = \mathrm{diag}(4, \dots, 4),$ which corresponded to $\sigma_\beta=0.5.$ 
For each $p$ we conducted $1000$ independent experiments.

Our initial studies compared ICF with six other methods: the Cholesky decomposition (Chol), the Jacobi method, the Gauss-Seidel method (GS), the successive over-Relaxation (SOR) method, the steepest descent method, and the conjugate gradient (CG) method.  
We excluded the Jacobi method and the steepest descent method due to their poor performance. 
To ensure a fair comparison in the context of BVSR, the starting point for Chol, GS, and SOR was $\bA = \bX^t \bX + \sigma^{-2}_\beta \bI$ being obtained, and the starting point for ICF was the upper triangular matrix $\bR$ such that $\bX^t \bX = \bR^t \bR$ being obtained.  
For GS, we tried the preconditioning method of~\citet{kohno1997improving}, which is the most efficient among the methods surveyed in~\citet{niki2004survey}, but we observed no improvement, most likely because $\bA$ is well-conditioned due to the regularization.
For SOR, we need choose a value for the relaxation parameter (denoted by $\omega_{\mathrm{SOR}}$), which is known to be very difficult. A solution was provided by~\citet{young1954iterative}~\citep[see also][]{yang2007optimal}, 
but we observed that it did not apply when $p>500$.
After trial and error, we settled on using $\omega_{\mathrm{SOR}} = 1.2$, which appeared to be optimal in our numerical studies. 
For all iterative methods, we started from $\hat{\bbeta}^{(0)} = \bzero$ and stopped if 
\begin{equation}\label{eq:stop}
\max\limits_{j} |  \hat{\bbeta}^{(k)}_j - \hat{\bbeta}^{(k-1)}_j  |  < 10^{-6},
\end{equation}
or the number of iterations exceeded $200$. 
The computer code for different methods was written in C++ and was run in the same environment.  
The Cholesky decomposition was implemented using GSL (GNU Scientific Library)~\citep{gsl}, and 
GS and SOR were implemented in a manner that accounted for the sparsity of the triangular matrices $\bL$ and $\bU$ to obtain maximum efficiency~\citep[Chap.~10.1.2]{golub2012matrix}. 
Lastly, we included in the comparison the LAPACK routine DGELS, the most widely used least square solver, as a baseline reference.

\begin{table}[h!]
\centering
\noindent
{\small
\begin{tabular}{clllllllllll } \toprule
\multirow{2}{*}{Dataset} & \multirow{2}{*}{$p$} & \multicolumn{6}{c}{Time (in seconds)} & \multicolumn{4}{c}{Convergence failures}  \\ 
\cmidrule(lr){3-8}
\cmidrule(lr){9-12}
&  &   {\footnotesize DGELS} &  {\footnotesize Chol} & {\footnotesize ICF} & {\footnotesize GS} & {\footnotesize SOR} & {\footnotesize CG} & {\footnotesize ICF}   & {\footnotesize GS}   & {\footnotesize SOR}  & {\footnotesize CG}  \\ 
\midrule
\multirow{5}{*}{IND}&50  & 7.3 & 0.034 & 0.019 & {\bf 0.016} & 0.025 & 0.054 & 0 & 0 & 0 & 0 \\
&100 & 26 &  0.20 & {\bf 0.07} & {\bf 0.07} & 0.09 & 0.24 & 0 & 0 & 0 & 0 \\
&200 & 115 &  1.38 & {\bf 0.30} & 0.34 & 0.39 & 1.15 & 0 & 1 & 1 & 0 \\
&500 & 679 & 21.0 & {\bf 2.7} & 3.6 & 2.9 & 9.8 & 0 & 2 & 1 & 0 \\
&1000 & 2948 & 161 & {\bf 14} & 36 & 25 & 60 & 0 & 11 & 7 & 0 \\
\midrule
\multirow{5}{*}{DEP} &50  & 7.3 &  0.035 & {\bf 0.020} & 0.029 & 0.031 & 0.056 & 0 & 8 & 6 & 0 \\
&100 & 27 & 0.20 & {\bf 0.08} & 0.23 & 0.19 & 0.26 & 0 & 29 & 20 & 0 \\
&200 & 116 & 1.39 & {\bf 0.45} & 2.13 & 1.69 & 1.34 & 0 & 125 & 93 & 0 \\
&500 & 683 & 21.0 & {\bf 5.5} & 36.1& 32.8 & 15.1 & 0 & 621 & 514 & 0 \\
&1000 & 2984 &  160 & {\bf 36} & 183 & 180 & 133 & 0 & 979 & 951 & 0 \\
\bottomrule
\end{tabular}
}

\caption{Wall time usage (in seconds) and numbers of convergence failures.  The top half is for the IND dataset and the bottom half is for the DEP dataset.  The statistics for each of the six methods were obtained from $1,000$ independent repeats. 
DGELS: LAPACKE\_dgels routine; 
Chol: Cholesky decomposition; ICF: iterative  complex factorization; GS: Gauss-Seidel method; SOR: successive over-relaxation; CG: conjugate gradient.
``Convergence failures" columns give the numbers of experiments that fail to converge within $200$ iterations for the four iterative methods. 
}\label{table:time2}
\end{table}

The results are summarized in Table~\ref{table:time2}. For the IND dataset, three iterative methods (ICF, GS and SOR) appear on par with each other for smaller $p$.  For $p=1000$, ICF outperforms the other two, and is 10 times faster than the Cholesky decomposition. 
On average it took ICF $5$ iterations to converge, and ICF never failed to converge in all experiments. 
On the other hand, both GS and SOR failed to converge, at least once, for large $p$.  
For the DEP dataset, we note that ICF is the fastest among all the methods compared. 
For larger $p$ $(p=500,1000)$, ICF is 4 -- 5 times faster than the Cholesky decomposition, and 5 -- 6 times faster than the other three methods. 
Both GS and SOR had difficulty in converging within $200$ iterations for large $p$. 
DGELS is always the slowest since it assumes $\bX^t \bX$ is unknown and solves~\eqref{eq:beta2} by QR decomposition. 
Advanced least square solvers~\citep{avron2010blendenpik, meng2014lsrn} have similar performance to DGELS and can only beat it by a small margin when $\bX$ is very large.
We also tweaked simulation conditions to check whether the results were stable. 
For example, we tried to simulate $\bz$ under the alternative instead of the null, and to initialize $\hat{\bbeta}^{(0)}$ in different manners, such as using unpenalized linear estimates. The results remained essentially unchanged under different tweaks.

\begin{figure}[htb!]
\begin{center}
\includegraphics[width=0.48\linewidth]{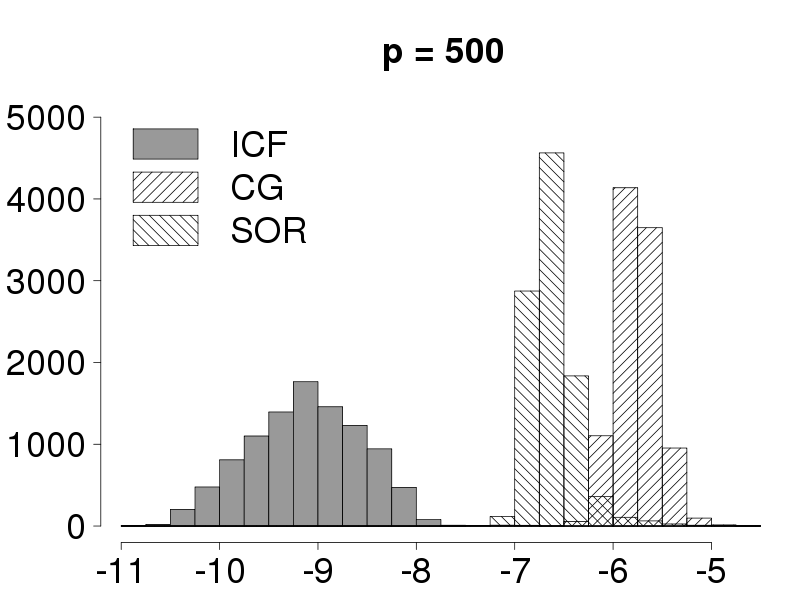} 
\includegraphics[width=0.48\linewidth]{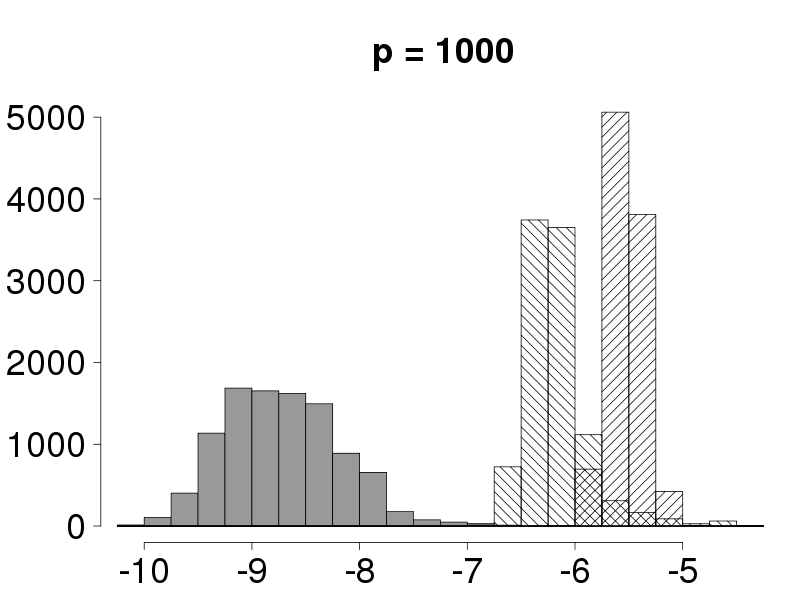}  
\caption{Comparison of the accuracy using the IND dataset. The three methods compared here are ICF (iterative  complex factorization), CG (conjugate gradient), and SOR (successive over-relaxation).  
Each panel shows the distributions of the maximum (entry-wise) absolute error on $\log_{10}$ scale. 
}\label{fig:acc}
\end{center}
\end{figure}

Next, we compared the accuracy of different iterative methods, where
the truth  was obtained from the Cholesky decomposition. 
In this experiment we used the IND dataset and compared for $p=500$ and $p=1000$. 
Each method was repeated for $10,000$ experiments. The maximum entry-wise deviation (denoted by $d$) was obtained for each method, each experiment, under each simulation condition. 
Only those converged experiments were included in the comparison. Figure~\ref{fig:acc} shows the distributions of $\log_{10}d$ for three methods. Clearly, ICF outperforms CG and SOR by a large margin. 
(GS is omitted since its accuracy is poorer than that of SOR in almost every experiment.) 
Because in real applications we are oblivious to the truth (or the truth is expensive to obtain), a method is more desirable if its deviation from the truth is within the pre-specified precision. 
Figure~\ref{fig:acc} shows that ICF always achieves the pre-specified precision ($10^{-6}$) while the other two methods do not. 
Moreover, the deviation of ICF is 2 -- 3  orders of magnitude smaller than that of SOR, and 3 -- 4 orders of magnitude smaller than that of the CG method. 
We repeated the experiments for the DEP dataset and made the same observations. 

\begin{figure}[htb!]
\begin{center}
\includegraphics[width=0.48\linewidth]{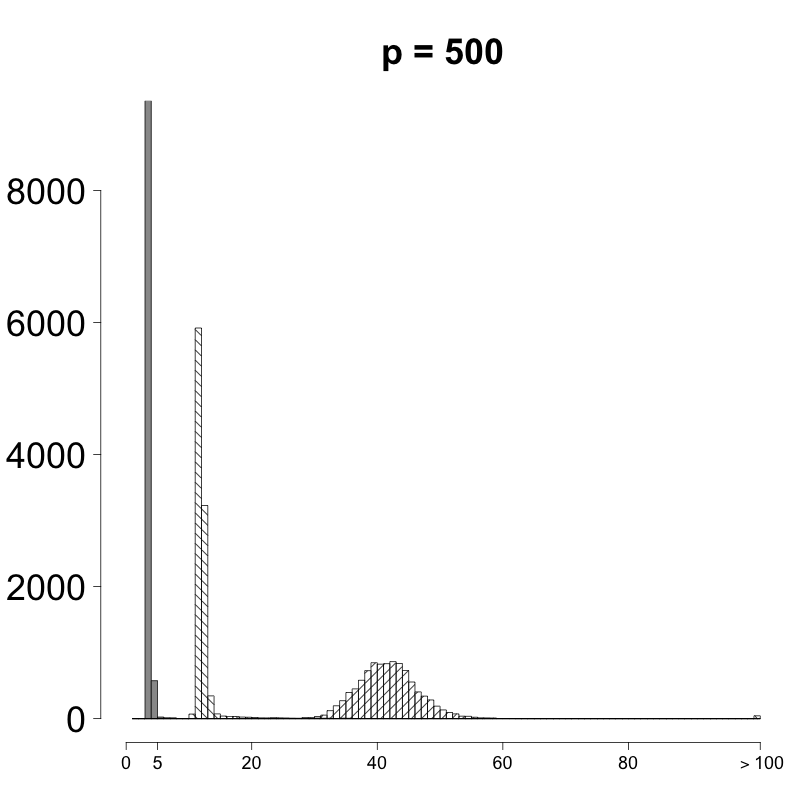} 
\includegraphics[width=0.48\linewidth]{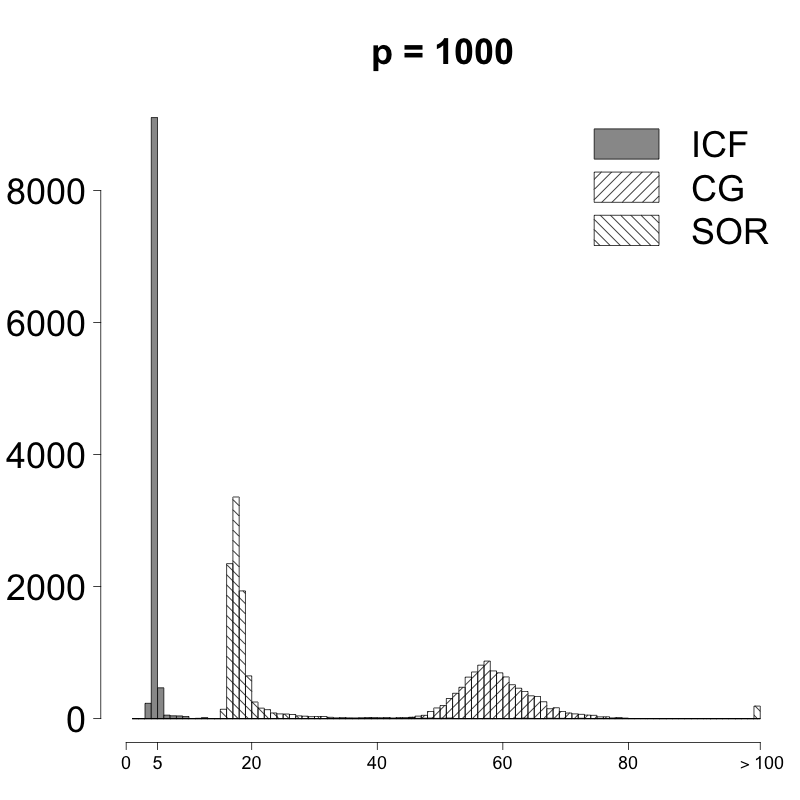}  
\caption{Comparison of the convergence speed using the IND dataset. 
Each panel shows the distributions of the number of iterations needed to stop by rule~\eqref{eq:stop}. 
}\label{fig:conv}
\end{center}
\end{figure}

Figure~\ref{fig:conv} compares the number of iterations used to converge for different iterative methods for large $p$. ICF takes a few iterations to converge, SOR a dozen iterations, CG about $50$ iterations. More numerical experiments can be found in Supplementary S2 where we investigated how $p$ and $\sigma_\beta$ affects convergence rates.
We also compared the performance of the iterative methods when the design matrix $X$ is beyond the counts of reference alleles, such as normal or log-normal distributed (Supplementary S1), and when $X$ has a different degree of collinearity (Supplementary S2.2).
ICF exhibits an overwhelming advantage in every scenario.

\section{ICF dramatically increases the speed of BVSR}
\subsection{Incorporating ICF into BVSR}\label{sec:bvsr.icf}
As we mentioned in the introduction, for BVSR the most time-consuming step in MCMC  is the computation of~\eqref{eq:bf} for every proposed $(\bgamma', h')$ (superscript $'$ denotes the proposed value). 
To successfully incorporate ICF into BVSR, we need to overcome two hurdles: 
avoiding the computation of the determinant term in~\eqref{eq:bf}  
and efficiently obtaining the Cholesky decomposition $\bR^t \bR = \bX^t_{\bgamma}\bX_{\bgamma}$. 

Computing matrix determinant has a cubic complexity. If this is not avoided, using ICF to compute $\hbb$ becomes pointless. 
When evaluating~\eqref{eq:bf}, we need to compute the determinant $|\bX^t_{\bgamma} \bX_{\bgamma} + \sigma_{\beta}^{-2}(\bgamma, h) \bI |,$ which takes $O(|\bgamma|^3)$ operations. 
Identifying that the determinant is a normalization constant, avoiding computation of the determinant becomes a well studied problem in the MCMC literature that deals with the so-called ``doubly-intractable distributions," i.e. distributions with two nested unknown normalization constants~\citep{moller2006efficient,murray2012mcmc}. 
Recall that we want to sample 
\begin{equation}
\begin{aligned}
\P( \bgamma  , h \mid \y) &\propto \P(\y  \mid  \bgamma, h) \P(\bgamma, h) \\
& \propto   |\bOmega(\bgamma, h)|^{-1/2}\;\sigma_\beta^{-|\bgamma|}(\bgamma, h)  \Biggl( 1 - \frac{ \y^t\X_{\bgamma} \bOmega^{-1}(\bgamma, h) \X_{\bgamma}^t \y}{\y^t \y -  n \bar{y}^2  } \Biggr)^{-n/2} \P(\bgamma, h),
\end{aligned}
\end{equation}
where   $\bOmega(\bgamma, h) =  \bX_{\bgamma}^t \bX_{\bgamma} + \sigma_\beta^{-2} (h,\bgamma)\I $, and it is the computation of $|\bOmega(\bgamma, h)|$ that we want to avoid.  
For a naive Metropolis-Hastings algorithm, if the current state is $(\bgamma, h)$ and the proposed move is $(\bgamma', h')$, we need to evaluate \emph{the Hastings ratio} 
\begin{equation}\label{eq:hastings}
\alpha= \frac{K(\bgamma, h \mid \bgamma', h')} {K(\bgamma', h'  \mid \bgamma, h)}
\frac{Z(\bgamma', h')}{Z(\bgamma, h)} \frac{L(\y, \bgamma', h')}{L(\y, \bgamma, h)} 
\frac{\P(\bgamma', h')} {\P(\bgamma, h)},
\end{equation}
where 
\begin{equation}
\begin{aligned}
Z(\bgamma, h) &=  |\bOmega(\bgamma, h)|^{-1/2} \sigma_\beta^{-|\bgamma| }(\bgamma, h), \\
L(\y, \bgamma, h ) &=  (\y^t\y - \y^t\bX_{\bgamma} \bOmega^{-1}(\bgamma, h) \bX_{\bgamma}^t \y  -  n \bar{y}^2  )^{-n/2}    , 
\end{aligned}
\end{equation}
and $K(\cdot' \mid \cdot)$ is the  
proposal distribution for proposing $\cdot'$ from $\cdot.$ 
The proposed move $(\bgamma', h')$ is accepted with probability $\min(1, \alpha),$ which is called \emph{the Metropolis rule}. Apparently, both $|\bOmega(\bgamma, h)|$ and $|\bOmega(\bgamma', h')|$ need to be evaluated in a naive MCMC implementation. 

To avoid computing $Z(\bgamma, h)$, notice that 
\begin{equation} \label{exchange}
\frac{Z(\bgamma', h')}{Z(\bgamma, h)} = \frac{\P(\tilde{\y} \mid \bgamma', h')}{\P(\tilde{\y} \mid \bgamma, h)} \frac{L( \tilde{\y}, \bgamma, h)}{L(\tilde{\y}, \bgamma', h')}
\end{equation}
holds for all $\tilde{\y}$, and if $\tilde{\y}$ is sampled from $\P(\cdot \mid \bgamma', h')$, the ratio ${L(\tilde{\y}, \bgamma, h)}/{L(\tilde{\y}, \bgamma', h')}$ becomes a one-sample importance-sampling estimate of ${Z( \bgamma', h')}/{Z( \bgamma, h)}$. 
Hence we can plug in  ${L(\tilde{\y}, \bgamma, h)}/{L(\tilde{\y}, \bgamma', h')}$ to replace ${Z( \bgamma' , h')}/{Z( \bgamma, h)}$  
and compute the Hastings ratio by 
\begin{equation}\label{eq:alpha2}
\alpha(\bgamma, h, \bgamma', h', \tilde{y})  = \frac{K(\bgamma, h \mid \bgamma', h')} {K(\bgamma', h'  \mid \bgamma, h)}
\frac{L( \tilde{\y}, \bgamma, h) }{L(\tilde{\y},  \bgamma', h')} \frac{L(\y, \bgamma', h')}{L(\y, \bgamma, h)} 
\frac{\P(\bgamma', h')} {\P( \bgamma, h)}. 
\end{equation}
This is the exchange algorithm of~\cite{murray2012mcmc}, which can be viewed as a special case of the pseudo-marginal method~\citep{andrieu2009pseudo}. 
In the Appendix we prove that the stationary distribution of this MCMC is the desired posterior distribution $\P( \bgamma , h \mid \y). $  
Since the Bayes factor~\eqref{eq:bf}  is invariant to the scaling and shifting of $\by$, we can simply assume $\mu = 0$ and $\tau = 1$ when sampling $\tilde{\by}$.  

To tackle the second difficulty, notice that in each MCMC iteration, usually only one or a few entries of $\bgamma$ are flipped between $0$ and $1$. Such proposals are often called ``local'' proposals, and Markov chains using local proposals usually have high acceptance ratios, which indicate the chains are ``mixing" well~\citep{guan.krone.07}. 
When one covariate is added into the model, the new Cholesky decomposition, $(\bR')^t \bR'$, can be obtained from the previous decomposition $\bR^t \bR$ by a forward substitution. 
When one covariate is deleted from the model, we simply delete the corresponding column from $\bR$ and introduce new zeros by Givens rotation~\citep[Chap.~5.1]{golub2012matrix}, which requires $\approx 3k^2$ flops if the $(|\bgamma| - k)$-th  predictor is removed. We provide a toy example in Appendix to demonstrate the update on the Cholesky decomposition.

\subsection{Numerical examples}
We developed a software package, fastBVSR, to fit the BVSR model by incorporating ICF and the exchange algorithm into the MCMC procedure. The algorithm is summarized in Supplementary S4. 
The software is written in C++ and available at \url{http://www.haplotype.org/software.html}. 
In fastBVSR, we also implemented Rao-Blackwellization, as described in~\citet{guan2011bayesian}, to reduce the variance of the estimates for $\bgamma$ and $\bbeta$.  
By default, Rao-Blackwellization is done every 1000 iterations. 

To check the performance of fastBVSR, we performed simulation studies based on a real dataset described in~\citet{yang2010common} (henceforth the Height dataset). The Height dataset contains $3,925$ subjects and $294,831$ common SNPs (minor allele frequency $\geq 5\%$) after routine quality control~\cite[c.f.][]{xu.guan.14}.  We sampled $10,000$ SNPs across the genome to perform simulation studies. 
Our aim was to check whether fastBVSR can reliably estimate the heritability in the simulated phenotypes. 
To simulate phenotypes of different heritability, we randomly selected $200$ causal SNPs (out of $10,000$) to obtain $\bgamma$.  For each selected SNP we drew its effect size from the standard normal distribution to obtain $\bbeta_{\bgamma}$ (the subvector of $\bbeta$ that contains nonzero entries). Then we simulated the standard normal error term to obtain $\bepsilon$, and scaled the simulated effect sizes simultaneously using $\lambda$ such that $\by = \lambda \bX_{\bgamma} \bbeta_{\bgamma}   + \bepsilon$ and the heritability, $1 - \Var(\bepsilon)/ \Var(\by)$, was $h$ (taking values in $0.01, 0.02, \dots, 0.99$). This was the same procedure as that of~\cite{guan2011bayesian}. 
For each simulated phenotype, we ran fastBVSR to obtain the posterior estimate of heritability.  
We compared fastBVSR with GCTA~\citep{yang2011gcta}, which is a software package to estimate 
heritability and do prediction using the linear mixed model. 
For heritability estimation, GCTA has been shown to be unbiased and accurate in a wide range of settings~\citep{yang2010common,  lee2011estimating}. 
The result is shown in Figure~\ref{fig:herit1}. 
Both fastBVSR and GCTA can estimate the heritability accurately; the mean absolute error is $0.014$ for fastBVSR and $0.029$ for GCTA.
Noticeably, fastBVSR has a smaller variance but a slight bias when the true heritability is large. 

\begin{figure}[htbp!]
\begin{center}
\includegraphics[width=0.45\linewidth]{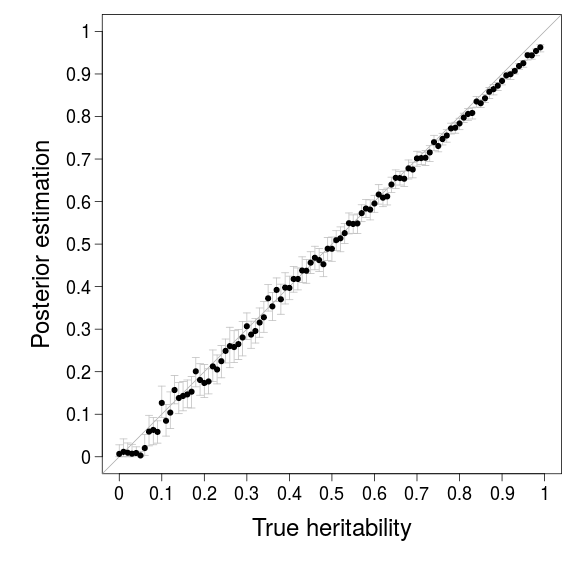}
\includegraphics[width=0.45\linewidth]{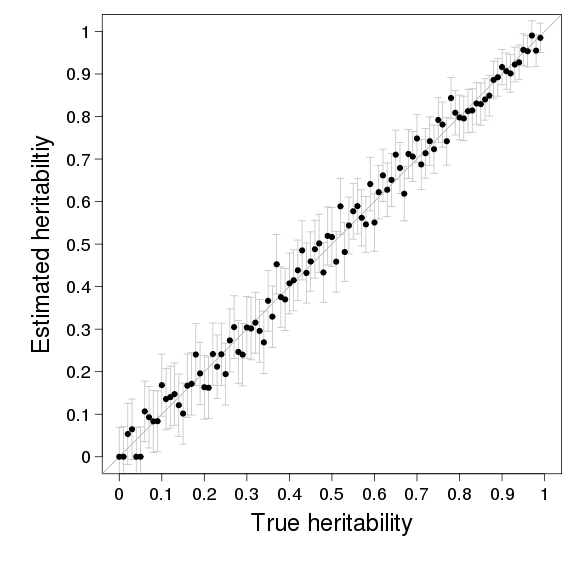}\\
\caption{Heritability estimates. The left panel is the result of fastBVSR, obtained using $20,000$ burn-in steps and $100,000$ sampling steps, and the right panel is the result of GCTA. 
In each iteration of BVSR, the heritability is estimated by computing the proportion of explained variance of $\by$ using the sampled parameter values. 
The grey bars represent $95$\% credible intervals for fastBVSR and $\pm 2$ standard error for GCTA. }\label{fig:herit1}
\end{center}
\end{figure}

Next we compare the predictive performance of fastBVSR and GCTA. 
We first define the mean squared error as a function of $\hat{\bbeta}$, 
\begin{equation*}
\mathrm{MSE} (\hat{\bbeta})\, \define \, \dfrac{1}{n} || \bX\bbeta - \bX \hat{\bbeta}  ||^2_2.
\end{equation*}
Then following~\citet{guan2011bayesian} we define relative prediction gain (RPG) to measure the predictive performance of $\hat{\bbeta}$ 
\begin{equation*}
\mathrm{RPG}  \, \define \,  \dfrac{ \mathrm{MSE}(\bm{0}) - \mathrm{MSE}( \hat{\bbeta} ) }{\mathrm{MSE}(\bm{0}) - \mathrm{MSE}(\bbeta) }.  
\end{equation*}
The advantage of RPG is that the scaling of $\by$ does not contribute to RPG so that simulations with different heritability can be compared fairly. Clearly when $\hat{\bbeta} = \bbeta$, $\mathrm{RPG} = 1$; when $\hat{\bbeta} = \bm{0}$, $\mathrm{RPG} = 0$.  
Figure~\ref{fig:rpg} shows that fastBVSR has much better predictive performance than GCTA, which reflects the advantage of BVSR over the linear mixed model. 
This advantage owes to the model averaging used in BVSR~\citep{raftery1997bayesian, broman2002model}. Besides, note that BVSR with Rao-Blackwellization performs slightly better than BVSR alone. 

\begin{figure}[htbp!]
\begin{center}
\includegraphics[width=0.5\linewidth]{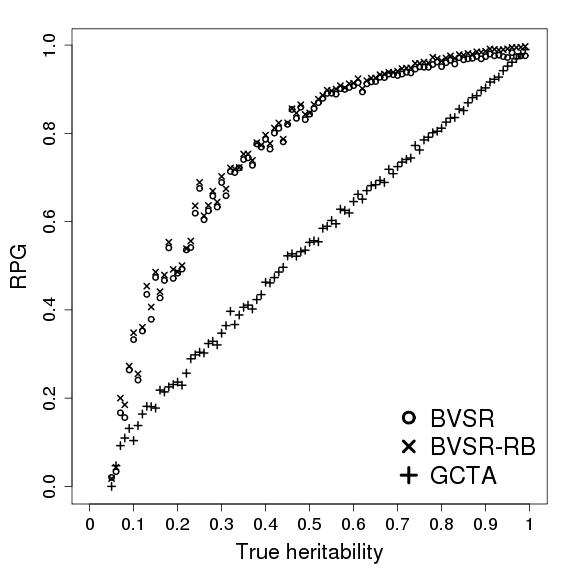}
\caption{Relative prediction gain (RPG) of fastBVSR and GCTA. 
BVSR stands for the RPG of the crude posterior mean estimates from fastBVSR; BVSR-RB represents the RPG of the Rao-Blackwellized estimates from fastBVSR; for GCTA, RPG is computed using the BLUPs (best linear unbiased predictors) from linear mixed model. 
}\label{fig:rpg}
\end{center}
\end{figure}

Lastly we check the calibration of the posterior inclusion probability (PIP), $\P(\gamma_j = 1 \mid \by)$. We pool the results of the $99$ simulation sets with different heritability (from $0.01$ to $0.99$ at increment of $0.01$).  In each set there are $10,000$ estimated PIPs, one for each SNP, and in total there are $990,000$ PIP estimates. 
We group these PIPs into $20$ bins, $[0.05 \times (i-1), 0.05 \times i)$ for $i = 1, \dots, 20$, and for each bin we compute the fraction of true positives. If the PIP estimates are well calibrated, we expect that in each bin, the average of the PIPs roughly equals the fraction of true positives. 
Figure~\ref{fig:calib} shows that the PIP estimated by BVSR is conservative, which agrees with the previous study~\citep{guan2011bayesian}, and the PIP estimated by Rao-Blackwellization is well cablibrated. 

\begin{figure}[ht!]
\begin{center}
\includegraphics[width=0.45\linewidth]{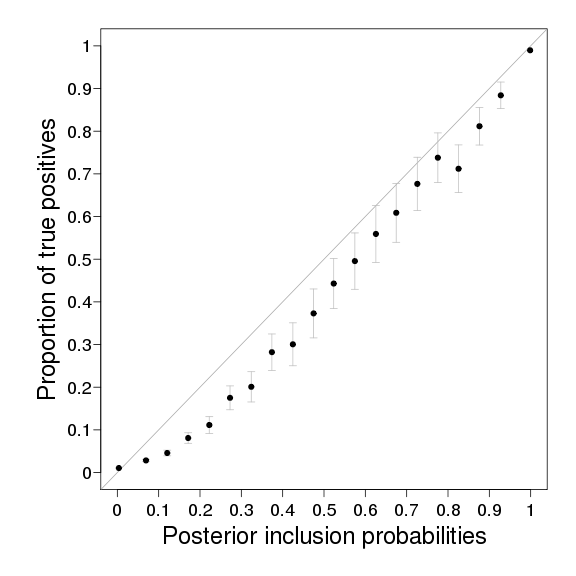}
\includegraphics[width=0.45\linewidth]{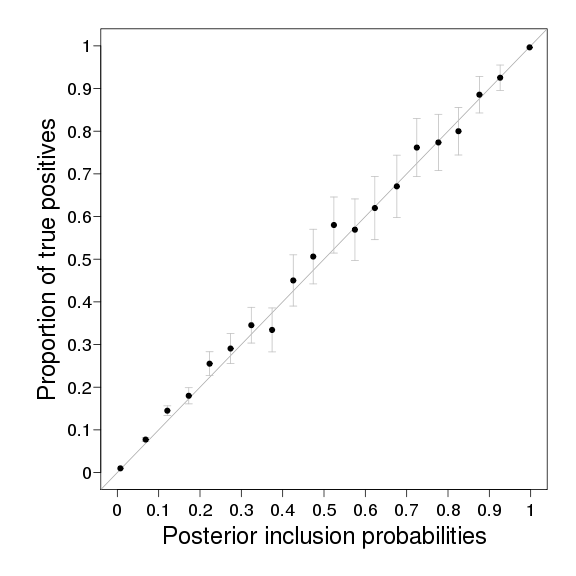}
\caption{Calibration of the posterior inclusion probability (PIP). 
In the left panel PIP is estimated directly from MCMC samples, and in the right panel PIP is estimated using Rao-Blackwellization to reduce variance. 
All the PIP estimates are pooled into 20 bins (see the main text for details). 
X-axis denotes the mean PIP of each bin and the Y-axis denotes the fraction of true positives in each bin. 
The grey bars represent $\pm 2$ standard error of the estimated proportion of true positives in each bin.  
The grey line is $y = x$. 
}\label{fig:calib}
\end{center}
\end{figure}

\subsection{Analysis of real GWAS datasets}\label{sec:iop}
We applied fastBVSR to analyze GWAS datasets concerning intraocular pressure (IOP).  We applied and downloaded two GWAS datasets from the database of Genotypes and Phenotypes (dbGaP), one for glaucoma (dbGaP accession number: phs000238.v1.p1) and the other for intraocular hypertension (dbGaP accession number: phs000240.v1.p1). 
The combined dataset contains $301,143$ autosomal SNPs and $3,226$ subjects, and was used previously by~\citet{zhou.guan.jasa} to examine the relationship between p-values and Bayes factors.

We conducted five independent runs with different random seeds. 
The sparsity parameters for BVSR were set as $\pi_{min} = 0.0001$ and $\pi_{max} = 0.01$, which reflected a prior of 30 -- 3000 marginally associated SNPs. 
The posterior mean model size in each run ranges from $430$ to $622$~(Table~\ref{table:iop}). 
The average number of iterations for ICF to converge is $7.3$ ($95$\% interval = $[4, 13]$), which suggests that ICF is effective in analyzing real datasets. 
The speed improvement of fastBVSR over piMASS (a BVSR implementation that uses the Cholesky decomposition to fit Bayesian linear regression, a companion software package of~\cite{guan2011bayesian}) is dramatic: with model size  around $500$ and more than $3000$ individuals, fastBVSR took $14$ hours to run $2.1$ million MCMC iterations, compared to more than $1000$ hours used by piMASS on problems of matching size. 
Out of concern for the cumulative error in obtaining the matrix $\bR$ by updating the Cholesky decomposition (detailed in Section~\ref{sec:bvsr.icf}), we compared every $1000$ steps the updated Cholesky decomposition with the directly computed Cholesky decomposition, and results suggested that the entry-wise difference was absolutely negligible.  

We examine the inference of the hyperparameters, namely, the narrow-sense heritability $h$, the model size $|\bgamma|$, and the sparsity parameter $\pi$. 
Table~\ref{table:iop} shows the inferences of the three parameters are more or less consistent across five independent runs,  and thus we combine the results from five independent runs in the following discussion.     
The posterior mean for $h$ is $0.28$ with $95\%$ credible interval $(0.1, 0.44).$ This estimate of the heritability is smaller than those reported in the literature:  $0.62$ from a twin study~\citep{carbonaro2008heritability}, $0.29$ from a sibling study~\citep{chang2005determinants}, and $0.35$ from an extended pedigree study~\citep{van2007genetic}. For comparison, using the same dataset, the heritability estimated by GCTA is $0.47$ with standard error $0.11$.  The underestimation of $h$ using fastBVSR is perhaps due to over-shrinking of the effect sizes (the posterior mean of prior effect size $\sigma_\beta$ is $0.048$). 
The posterior mean for $\pi$ is $0.0016$, which suggests that IOP is a very much polygenic phenotype, agreeing with the complex etiology of glaucoma~\citep{glaucoma.14}, to which high IOP is a precursor phenotype.

\begin{table}[h!]
\centering
\noindent
\begin{tabular}{ cllllll } \toprule
  &  Combined    & Run1 & Run2 & Run3 & Run4 & Run5\\
$h  $  &  0.28  & 0.27 & 0.26 & 0.32 & 0.28 & 0.26 \\ 
$|\bgamma| $ & 474 &  446 &  395 &  622 &  476 & 430  \\
$\pi$ & 0.0016 &  0.0015 &  0.0013 &  0.0021 & 0.0016 & 0.0014 \\
\bottomrule
\end{tabular}
\caption{Posterior mean estimates for the hyperparameters: heritability $h$, model size $|\bgamma|$, and sparsity parameter $\pi$.  
In each MCMC iteration, heritability is re-estimated by computing the proportion of explained variance of $\by$ using the sampled parameters. 
The first column gives the combined results from the five independent runs; each run has $2.1$ million MCMC steps, including $0.1$ million burn-in steps. 
}\label{table:iop}
\end{table}

Lastly, we examine the top marginal signals detected by fastBVSR. Table~\ref{table:iop2} lists top $11$ SNPs based on PIP inferred from fastBVSR, at an arbitrary PIP threshold of $0.25$. Table~\ref{table:iop2} also includes PPA (posterior probability of association) based on the single SNP test for two choices of prior effect sizes: $\sigma_\beta = 0.05$ which is the posterior mean of $\sigma_\beta$ inferred from fastBVSR and $\sigma_\beta=0.5$ which is a popular choice for single SNP analysis.  We use the posterior mean $\pi = 0.0016$ inferred from fastBVSR as the prior odds of association. Multiplying a Bayes factor and the prior odds we obtain the posterior odds, from which we obtain the $\PPA = (posterior\;odds) / (1 + posterior \;odds).$  
First the two columns of PPAs are highly similar in both ranking and magnitude. This observation agrees with our experience and theoretical studies~\citep{zhou.guan.jasa} that modest prior effect sizes produce similar evidence for association. Second, PPA and PIP have significantly different rankings (Wilcox rank test p-value $= 0.025$). There are two related explanations for this observation: if two SNPs are correlated, the PIPs tend to split between the two; conditioning on the other associated SNPs, the marginal associations tend to change significantly.  
Three known GWAS signals, \emph{rs12150284, rs2025751, rs7518099/rs4656461}, which were replicated in our previous single-SNP analysis~\citep{zhou.guan.jasa},  all show appreciable PIPs. 
Moreover, we note a potential novel association first reported in~\citet{zhou.guan.jasa}:  
two SNPs located in \emph{PEX14}, \emph{rs12120962,  rs12127400}, have PIPs of $0.57$ and $0.27$ respectively.

\begin{table}[h!]
\centering
\noindent
\begin{tabular}{ lllllll } \toprule
SNP &  Chr   & Pos (Mb) & $\PPA_{\sigma_\beta=0.5}$ &$\PPA_{\sigma_\beta=0.05}$ &  $\PIP_{\mathrm{BVSR}}$ \\
{\bf rs12120962}   & 1 & 10.53  & 0.85  & 0.84  & 0.57  \\
{\bf rs12127400}   & 1 & 10.54  & 0.74  & 0.74   & 0.27  \\
{\bf rs4656461}   & 1 & 163.95  & 1.0 &  0.96 & 0.36  \\
{\bf rs7518099}   & 1 & 164.00  & 1.0  & 0.98 &  0.42  \\
rs7645716          & 3 & 46.31    & 0.62 & 0.56 &  0.35  \\
{\bf rs2025751}    & 6 & 51.73   & 0.81  & 0.82   & 0.68  \\
rs10757601         & 9 & 26.18    & 0.47  & 0.55  & 0.31  \\
rs9783190        & 10 & 106.76     & 0.39  & 0.31  & 0.32  \\
rs1381143         & 12 & 86.76  & 0.27 & 0.37  & 0.31  \\
rs4984577         & 15 & 93.76  & 0.43  & 0.49   & 0.42  \\
{\bf rs12150284}  & 17 & 9.97    & 0.99  & 0.97  & 0.92  \\
\bottomrule
\end{tabular}
\caption{SNPs with (Rao-Blackwellized) $\PIP >0.25$. Chr is short for chromosome, and Pos is short for position, which  is measured in mega-basepair (Mb) according to HG18. 
$\PPA$ stands for the posterior inclusion probability. Computing PPA requires the Bayes factor and the prior odds.  For prior odds we use $0.0016$, which is the posterior mean of $\pi$ inferred from fastBVSR. Bayes factors are computed using two priors effect sizes ($\sigma_\beta=0.05, 0.5$).  
SNPs that are mentioned in the main text are highlighted in bold. 
}\label{table:iop2}
\end{table}

\section{Discussion}\label{sec:disc}
We developed a novel algorithm, iterative complex factorization, to solve a class of penalized linear systems, proved its convergence, and demonstrated its effectiveness and efficiency through simulation studies.  The novel algorithm can dramatically increase the speed of BVSR, which we demonstrated by analyzing a real dataset. 
In our simulation studies we only considered $n > p$ 
(in BVSR $p = |\bgamma|$ is the size of the selected model). 
When $n \le p$, the ICF can be implemented with the dual variable method of~\citet{saunders1998ridge}~\citep[see also][]{lu2013faster}, which hinges on the identity  
$(\bX^t_{\bgamma} \bX_{\bgamma} + \sigma^{-2}_\beta \bI)^{-1} \bX^t_{\bgamma} \by = \bX^t_{\bgamma} ( \bX_{\bgamma} \bX^t_{\bgamma} + \sigma^{-2}_\beta \bI )^{-1} \by$.

One limitation of ICF is that it requires the Cholesky decomposition to obtain the effective complex factorization. Nevertheless, ICF still has a range of potential applications.  The general ridge regression~\citep{draper1979ridge} is an obvious one, and here we point out two other examples. The first is the Bayesian sparse linear mixed model (BSLMM) of~\citet{zhou2013polygenic}, which can be viewed as a generalization of BVSR. BSLMM has two components, one is the linear component that corresponds to $\bX_{\bgamma}$ and the other is the random effect component.  The linear component requires one to solve a system that is similar to~\eqref{eq:beta2}~\citep[Text S2][]{zhou2013polygenic}; therefore ICF can be applied to BSLMM to make it more efficient. 
Another example of ICF application is using the variational method to fit the BVSR~\citep{carbonetto2012scalable,huang2016variational}. 
In particular, \citet{huang2016variational} estimated  $(\bgamma, \bbeta)$ using an iterative method.  
In each iteration, the posterior mean for $\bbeta$  is updated by solving a linear system that has the same form as~\eqref{eq:beta2}, where the dimension is equal to the number of covariates with nonzero PIP estimates,  and the diagonal matrix $\bSigma$ depends on the PIP estimates.  
Applying ICF to the variational method might be more beneficial compared to BVSR because the dimension of the linear system is exceedingly larger in the variational method.

The iterative complex factorization method converges almost instantaneously, and it is far more accurate than other iterative algorithms such as the Gauss-Seidel method. Why does it converge so fast and at the same time so accurate? 
Our intuition is that the imaginary part of the update in~\eqref{eq:icf.iter2} , which involves the skew-symmetric matrix $\bS=\bR^t\bSigma -\bSigma \bR$, whose Youla decomposition~\citep{youla} suggests that it rotates $\hbb^{(k)}$, scales and flips its coordinates and rotates back, and which then appends itself to be the imaginary part of $\bz$, provides a ``worm hole" through the imaginary dimension in an otherwise purely real, unimaginary, solution path. Understanding how this works mathematically will have a far-reaching effect on Bayesian statistics, statistical genetics, machine learning, and computational biology.

Finally, we hope our new software fastBVSR, which is 10 -- 100 times faster than the previous implementation of BVSR, will facilitate the wide-spread use of BVSR in analyzing or reanalyzing datasets from genome-wide association (GWA) and expression quantitative trait loci (eQTL) studies.

\section*{Supplementary material}
The online appendix includes additional numerical studies (S1 -- S3), a summary of the MCMC algorithm (S4), examples for updating the Cholesky decomposition (S5) and a proof for the stationary distribution of the exchange algorithm (S6). 
S1 compares the performance of the iterative methods using simulated datasets; S2 studies the behavior of the convergence rates of the iterative methods; S3 provides numerical evidence for the assumption $\eta_{min}^2 \approx 0$ used in Section~\ref{sec:icf.conv}. 

\bibliographystyle{ba}
\bibliography{reference}

\begin{acknowledgement}
We thank the two anonymous reviewers for their helpful comments that improved the quality of our presentation.
\end{acknowledgement}

\newpage
\pagestyle{empty}
\newcommand{\specialcell}[2][c]{\begin{tabular}[#1]{@{}c@{}}#2\end{tabular}}
\renewcommand{\thesection}{S\arabic{section}}   
\renewcommand{\thetable}{S\arabic{table}}
\renewcommand{\thefigure}{S\arabic{figure}}
\renewcommand{\thesubsection}{S\arabic{section}.\arabic{subsection}}%
\setcounter{table}{0}
\setcounter{section}{0}
\setcounter{subsection}{0}
\setcounter{figure}{0}

\begin{center}
{\bf\Large Fast model-fitting of Bayesian variable selection regression using the iterative complex factorization algorithm (Supplementary)} \\
\vspace{0.3cm}
{\large Quan Zhou and Yongtao Guan} 
\end{center}

\section{Numerical comparison studies of the iterative methods using simulated datasets}\label{sec:supp.sim1}
In Section~\ref{sec:sim.icf}, we performed numerical comparison studies using real GWAS datasets, where each covariate follows a binomial distribution.  
For comparison, we conducted the same experiment using simulated datasets where $X_{ij}$ is continuous variable. 
Let $X_{(i)}$ be an arbitrary row of $\bX$.  
First, we sampled $X_{(i)}$ from a multivariate normal distribution with fixed pairwise correlation $r$; that is, $X_{(i)}$ is an independent sample from $\MVN(0, \bC_r )$ where $(\bC_r)_{ij} = 1$ if $i = j$ and $(\bC_r)_{ij} = r$ otherwise. 
Second, we sampled $X_{(i)}$ from a multivariate log-normal distribution such that 
$\log X_{(i)} \sim \MVN(0, \bC_r) $.  
For both cases, we tried $r = 0, 0.2, 0.5, 0.9$. 
Note that for the log-normal case, the pairwise correlation is then equal to $0, 0.13, 0.38, 0.85,$ respectively. 
Then we applied the same procedure and parameter values described in Section~\ref{sec:sim.icf} ($n = 3000$ and $\sigma_\beta = 0.5$) and collected the wall time usage of the iterative methods for solving $1,000$ independent linear systems with form~\eqref{eq:beta2}. 
The only difference is that we used the relaxation parameter $\omega_{\mathrm{SOR}} = 0.4$ for the successive over-relaxation method since it appeared to produce best overall performance. 
The results are summarized in Table~\ref{table:supp.time}. 
In almost every scenario, ICF exhibits an overwhelming advantage, especially when the data is heavy-tailed (the log-normal case). 
The only exception is the normal data with $r = 0.9$, where CG also works well. 
Compared with the results given in Section~\ref{sec:sim.icf}, the advantage of ICF over the Cholesky decomposition becomes more prominent.   
Lastly, we note that GS and SOR work poorly when the pairwise correlation $r \geq 0.2$, and as $r$ grows larger, they quickly become ineffective. This phenomenon will be discussed in the next section.

\begin{table}[htbp!]
\centering
\noindent
{\small
\begin{tabular}{cllllllllll } \toprule
\multirow{2}{*}{Dataset} & \multirow{2}{*}{$p$} & \multicolumn{5}{c}{Time (in seconds)} & \multicolumn{4}{c}{Convergence failures}  \\ 
\cmidrule(lr){3-7}
\cmidrule(lr){8-11}
&  &   {\footnotesize Chol} & {\footnotesize ICF} & {\footnotesize GS} & {\footnotesize SOR} & {\footnotesize CG} & {\footnotesize ICF}   & {\footnotesize GS}   & {\footnotesize SOR}  & {\footnotesize CG}  \\ 
\midrule
\multirow{5}{*}{\specialcell{Normal \\ $r = 0$}} & 50  & 0.031 & {\bf 0.017} & {\bf 0.017} & 0.063 & 0.023 & 0 & 0 & 0 & 0\\ 
 & 100  & 0.19 & {\bf 0.06} & 0.06 & 0.25 & 0.09 & 0 & 0 & 0 & 0\\ 
 & 200  & 1.35 & {\bf 0.25} & 0.29 & 1.08 & 0.38 & 0 & 0 & 0 & 0\\ 
 & 500  & 20.6 & {\bf 1.8} & 2.7 & 9.7 & 3.2 & 0 & 0 & 0 & 0\\ 
 & 1000  & 159 & {\bf 7.3} & 21 & 72 & 21 & 0 & 0 & 0 & 0\\ 
\midrule
\multirow{5}{*}{\specialcell{Normal \\ $r = 0.2$}} & 50  & 0.032 & {\bf 0.017} & 0.180 & 0.089 & 0.025 & 0 & 0 & 0 & 0\\ 
 & 100  & 0.19 & {\bf 0.06} & 1.82 & 0.61 & 0.10 & 0 & 525 & 0 & 0\\ 
 & 200  & 1.36 & {\bf 0.25} & 7.21 & 6.19 & 0.38 & 0 & 1000 & 31 & 0\\ 
 & 500  & 20.6 & {\bf 1.8} & 44.6 & 44.5 & 3.6 & 0 & 1000 & 1000 & 0\\ 
 & 1000  & 159 & {\bf 7.7} & 205 & 204 & 23 & 0 & 1000 & 1000 & 0\\ 
\midrule
\multirow{5}{*}{\specialcell{Normal \\ $r = 0.5$}} & 50  & 0.032 & {\bf 0.018} & 0.527 & 0.368 & 0.026 & 0 & 1000 & 0 & 0\\ 
 & 100  & 0.19 & {\bf 0.07} & 1.90 & 1.91 & 0.09 & 0 & 1000 & 1000 & 0\\ 
 & 200  & 1.35 & {\bf 0.26} & 7.21 & 7.22 & 0.46 & 0 & 1000 & 1000 & 0\\ 
 & 500  & 20.6 & {\bf 2.0} & 44.6 & 44.5 & 3.2 & 0 & 1000 & 1000 & 0\\ 
 & 1000  & 159 & {\bf 10} & 195 & 195 & 20 & 0 & 1000 & 1000 & 0\\ 
 \midrule
 \multirow{5}{*}{\specialcell{Normal \\ $r = 0.9$}} & 50  & 0.033 & {\bf 0.027} & 0.527 & 0.530 & {\bf 0.027} & 0 & 1000 & 1000 & 0\\ 
 & 100  & 0.19 & {\bf 0.10} & 1.90 & 1.91 & {\bf 0.10} & 0 & 1000 & 1000 & 0\\ 
 & 200  & 1.35 & 0.42 & 7.21 & 7.22 & {\bf 0.40} & 0 & 1000 & 1000 & 0\\ 
 & 500  & 20.7 & {\bf 3.0} & 44.6 & 44.5 & 3.5 & 0 & 1000 & 1000 & 0\\ 
 & 1000  & 159 & {\bf 12} & 190 & 189 & 18 & 0 & 1000 & 1000 & 0\\ 
 \midrule
 \multirow{5}{*}{\specialcell{Log-normal \\ $r = 0$}} & 50  & 0.032 & {\bf 0.011} & 0.015 & 0.058 & 0.026 & 0 & 0 & 0 & 0\\ 
 & 100  & 0.19 & {\bf 0.04} & 0.06 & 0.23 & 0.10 & 0 & 0 & 0 & 0\\ 
 & 200  & 1.35 & {\bf 0.17} & 0.28 & 1.00 & 0.43 & 0 & 0 & 0 & 0\\ 
 & 500  & 20.6 & {\bf 1.7} & 2.5 & 8.9 & 3.5 & 0 & 0 & 0 & 0\\ 
 & 1000  & 159 & {\bf 7.4} & 21 & 69 & 23 & 0 & 0 & 0 & 0\\ 
\midrule
\multirow{5}{*}{\specialcell{Log-normal \\ $r = 0.2$}} & 50  & 0.032 & {\bf 0.016} & 0.091 & 0.072 & 0.030 & 0 & 0 & 0 & 0\\ 
 & 100  & 0.19 & {\bf 0.06} & 0.84 & 0.34 & 0.12 & 0 & 0 & 0 & 0\\ 
 & 200  & 1.37 & {\bf 0.25} & 7.18 & 2.84 & 0.56 & 0 & 958 & 0 & 0\\ 
 & 500  & 20.7 & {\bf 1.8} & 44.5 & 44.4 & 4.9 & 0 & 1000 & 980 & 0\\ 
 & 1000  & 159 & {\bf 7.6} & 205 & 205 & 37 & 0 & 1000 & 1000 & 0\\ 
\midrule
\multirow{5}{*}{\specialcell{Log-normal \\ $r = 0.5$}} & 50  & 0.033 & {\bf 0.018} & 0.519 & 0.196 & 0.034 & 0 & 815 & 0 & 0\\ 
 & 100  & 0.19 & {\bf 0.06} & 1.90 & 1.80 & 0.17 & 0 & 1000 & 494 & 0\\ 
 & 200  & 1.36 & {\bf 0.25} & 7.21 & 7.22 & 0.76 & 0 & 1000 & 1000 & 0\\ 
 & 500  & 20.5 & {\bf 1.8} & 44.6 & 44.5 & 7.4 & 0 & 1000 & 1000 & 0\\ 
 & 1000  & 159 & {\bf 7.5} & 195 & 194 & 56 & 0 & 1000 & 1000 & 0\\ 
\midrule
\multirow{5}{*}{\specialcell{Log-normal \\ $r = 0.9$}} & 50  & 0.032 & {\bf 0.018} & 0.527 & 0.530 & 0.044 & 0 & 1000 & 1000 & 0\\ 
 & 100  & 0.19 & {\bf 0.07} & 1.90 & 1.91 & 0.19 & 0 & 1000 & 1000 & 0\\ 
 & 200  & 1.35 & {\bf 0.32} & 7.22 & 7.23 & 0.96 & 0 & 1000 & 1000 & 0\\ 
 & 500  & 20.5 & {\bf 2.4} & 44.6 & 44.5 & 9.3 & 0 & 1000 & 1000 & 0\\ 
 & 1000  & 159 & {\bf 13} & 203 & 202 & 74 & 0 & 1000 & 1000 & 0\\ 
\bottomrule
\end{tabular}
}
\caption{  
Wall time usage (in seconds) and numbers of convergence failures with simulated data. 
For normal data, $r$ is the correlation between any two different covariates; for log-normal data, it is the correlation between the logarithms of two different covariates. 
The statistics for each of the five methods were obtained from $1,000$ independent repeats. 
Chol: Cholesky decomposition; ICF: iterative  complex factorization; GS: Gauss-Seidel method; SOR: successive over-relaxation; CG: conjugate gradient.
``Convergence failures" columns give the numbers of experiments that fail to converge within $200$ iterations for the four iterative methods. 
} 
\label{table:supp.time}
\end{table}

\section{Numerical studies on the convergence rates of the iterative methods}
\subsection{Symmetric Toeplitz systems}
In this section, we assume the covariance matrix of $\bX$ is a symmetric Toeplitz matrix, which we will define shortly, denoted by $\bT$. 
Instead of sampling $\bX$, we simply let $\bX^t \bX = n \bT$ and consider the linear system 
$$\bA \hbb = (n \bT + \sigma_\beta^{-2} \bI )\hbb = \bz.$$
We will use several typical choices of $\bT$ to study how the convergence rates of the iterative methods change with $p$ and $\sigma_\beta$. 
As shown by Proposition~\ref{prop1}, the convergence rate of ICF is given by the spectral radius of $\bPsi(\omega)$. 
For SOR, using the notation introduced in~\eqref{def:a}, the convergence rate is the spectral radius of $(\bD + \omega \bL)^{-1} \left[ (1-\omega) \bD - \omega \bU \right]$, which we denote by $\rho_{\rm SOR}(\omega)$. 
By letting $\omega = 1$, we get the convergence rate of GS, which is $\rho_{\rm GS} = \rho( ( \bD +\bL)^{-1} \bU )$. 
Let $\lambda_{\max}$ and $\lambda_{min}$ be the largest and smallest eigenvalue of $\bA$ and thus the condition number of $\bA$ is $\kappa(\bA) = \lambda_{\max}/\lambda_{\min}$. 
Then the convergence rate of CG can be bounded from above by $\bar{\rho}_{\rm CG} = (\sqrt{\kappa (\bA)} - 1)/(\sqrt{\kappa (\bA)} + 1)$. Note that we cannot compute the exact convergence rate of CG, and in fact, CG can be regarded as a direct method since it always converges within $p$ iterations~\citep[Lec.~38]{trefethen1997numerical}.

We denote a symmetric Toeplitz matrix by $\bT(a_0, \dots, a_{p-1})$, which satisfies $T_{ij} = a_{|i - j|}$, i.e.
\begin{align*}
\bT(a_0, \dots, a_{p-1}) = 
\begin{bmatrix}
a_0 & a_1 &  a_2 &  \cdots & \cdots  & a_{p-1} \\
a_1 & a_0 &  a_1 &  \ddots  &    & a_{p-2} \\
a_2 & a_1 & \ddots & \ddots &  \ddots &  \vdots \\
\vdots & \ddots & \ddots & \ddots & a_1 &  a_2 \\
a_{p - 2} &   & \ddots  & a_1  & a_0 & a_1 \\
a_{p - 1} &  \cdots & \cdots &  a_2 & a_1 & a_0 
\end{bmatrix}. 
\end{align*}
We consider two structured choices of $\bT$.
\begin{enumerate}[(i)]
\item $\bT_c(r) = \bT(1, -r, r, -r, r, \dots)$. For a regression problem, this is equivalent to $\bT(1, r, r, \dots, r)$ up to sign flipping of the regression coefficients.  
\item $\bT_e(\ell) = \bT(a_0, \dots, a_{p-1})$ where $a_k = \exp( -k/\ell )$ and $\ell > 0$.  
This can be seen as the covariance matrix of a discretized Ornstein-Uhlenbeck process. 
\end{enumerate}
Note that in this setup, since 
$$ \bA  = n \bT + \sigma_\beta^{-2} \bI  = n \left(  \bT + \dfrac{1}{n \sigma_\beta^2} \bI \right), $$
the convergence rates of the iterative methods only depend on $\bT$ and the product $ n \sigma_\beta^2 $. 
Hence we can simply let $n$ be an arbitrary positive constant, and we used $n = 1000$.  
However, for a real scatter matrix $\bX^t \bX$, one should be careful when $n < p$ since $\bX^t\bX$ is then rank deficient.
Such ill-conditioned systems will be studied in Section~\ref{sec:ill}. 

We first fixed $\sigma_\beta = 0.5$ and studied how the convergence rates change with $p$ using five covariance matrices: $\bT_c(0.01),  \bT_c(0.1), \bT_c(0.5), \bT_e(5), \bT_e(50)$.  
The results are summarized in Table~\ref{table:supp.conv1}. 
As expected, as $p$ grows larger,  the convergence rate becomes slower for every method.  
For the covariance matrix $\bT_e$, since $T_{ij}$ decreases exponentially as $|i - j |$ increases, the impact of larger values of $p$ is not   significant. 
Strikingly, ICF  has a much faster convergence rate than all the other methods, even if we simply choose $\omega = 1$. 
In contrast, the convergence rates of GS and SOR quickly approach $1$ as $p$ increases. 
This is probably because GS, SOR (and also Jacobi iteration) favor diagonally dominant matrices for $\bA$. Hence, if for each row of $\bA$, the sum of off-diagonal entries increases linearly with $p$, these methods would easily run into convergence difficulties. 

Next, we fixed $p = 200$ and studied how the convergence rates change with $\sigma_\beta$ using $\bT_c(0.1)$ and $ \bT_e(5)$. 
The results are summarized in Table~\ref{table:supp.conv2}. 
As expected, a smaller value of $\sigma_\beta$   makes $\bA$ diagonally more dominant, and as a result, GS and SOR   converge faster. 
However, ICF exhibits an opposite trend: as $\sigma_\beta$ increases, ICF converges faster. This is probably due to the use of the Cholesky decomposition of $\bX^t \bX$, which implies that if $\sigma_\beta \uparrow \infty$, ICF would converge immediately. 
Note that, because it is the product $n \sigma_\beta^2 $ that really matters,  if $p$ and $\sigma_\beta$ are fixed and $n$ increases, ICF will converge faster while the other methods become slower. 

\begin{table}[htbp!]
\centering
\noindent
{\small
\begin{tabular}{c lllllllll } 
\toprule
\specialcell{Covariance \\ matrix}  &  $p$ & $\sqrt{ \kappa(\bA) }$ &  $\omega^\star_{\rm ICF}$  & $\omega^\star_{\rm SOR}$ &  $\rho(\bPsi(1))$  &  $\rho(\bPsi(\omega^\star))$ &  $\rho_{\rm GS}$ &  $\rho_{\rm SOR}$  &  $\bar{\rho}_{\rm CG}$ \\ 
\midrule
\multirow{5}{*}{ \specialcell{ $\bT_{c}(0.01)$ } }  & 50 & 1.23 & 0.9999 & 0.95 & 2.3e-04 & 1.2e-04 & 0.104 & 0.080 & 0.102 \\ 
 & 100 & 1.42 & 0.9997 & 0.90 & 6.0e-04 & 3.0e-04 & 0.214 & 0.171 & 0.172 \\ 
 & 200 & 1.74 & 0.9993 & 0.81 & 0.0013 & 6.5e-04 & 0.395 & 0.321 & 0.269 \\ 
 & 500 & 2.46 & 0.999 & 0.61 & 0.0027 & 0.0014 & 0.700 & 0.576 & 0.421 \\ 
 & 1000 & 3.33 & 0.998 & 0.42 & 0.0040 & 0.0020 & 0.875 & 0.745 & 0.538 \\ 
\midrule
\multirow{5}{*}{ \specialcell{$\bT_{c}(0.1)$} }  & 50 & 2.56 & 0.998 & 0.63 & 0.0032 & 0.0016 & 0.710 & 0.602 & 0.438 \\ 
 & 100 & 3.47 & 0.998 & 0.42 & 0.0046 & 0.0023 & 0.883 & 0.764 & 0.553 \\ 
 & 200 & 4.81 & 0.997 & 0.26 & 0.0060 & 0.0030 & 0.963 & 0.871 & 0.656 \\ 
 & 500 & 7.50 & 0.996 & 0.12 & 0.0078 & 0.0039 & 0.993 & 0.945 & 0.765 \\ 
 & 1000 & 10.6 & 0.996 & 0.06 & 0.0089 & 0.0044 & 0.998 & 0.972 & 0.827 \\ 
\midrule
\multirow{5}{*}{ \specialcell{$\bT_{c}(0.5)$} }  & 50 & 7.11 & 0.993 & 0.22 & 0.013 & 0.0066 & 0.985 & 0.939 & 0.753 \\ 
 & 100 & 10.0 & 0.992 & 0.12 & 0.016 & 0.0077 & 0.996 & 0.969 & 0.818 \\ 
 & 200 & 14.1 & 0.991 & 0.06 & 0.017 & 0.0087 & 0.9990 & 0.984 & 0.868 \\ 
 & 500 & 22.3 & 0.990 & 0.03 & 0.020 & 0.0097 & 0.9998 & 0.994 & 0.914 \\ 
 & 1000 & 31.5 & 0.990 & 0.01 & 0.021 & 0.010 & 0.99996 & 0.997 & 0.938 \\ 
\midrule
\multirow{5}{*}{ \specialcell{$\bT_{e}(5)$} }  & 50 & 9.51 & 0.984 & 1.06 & 0.032 & 0.016 & 0.747 & 0.719 & 0.810 \\ 
 & 100 & 9.74 & 0.984 & 1.07 & 0.032 & 0.016 & 0.772 & 0.748 & 0.814 \\ 
 & 200 & 9.81 & 0.984 & 1.08 & 0.032 & 0.016 & 0.792 & 0.774 & 0.815 \\ 
 & 500 & 9.84 & 0.984 & 1.05 & 0.032 & 0.016 & 0.804 & 0.797 & 0.815 \\ 
 & 1000 & 9.84 & 0.984 & 1.02 & 0.032 & 0.016 & 0.809 & 0.806 & 0.815 \\ 
\midrule
\multirow{5}{*}{ \specialcell{$\bT_{e}(50)$} }  & 50 & 51.4 & 0.838 & 1.00 & 0.387 & 0.162 & 0.965 & 0.965 & 0.962 \\ 
 & 100 & 64.1 & 0.837 & 1.06 & 0.388 & 0.163 & 0.970 & 0.970 & 0.969 \\ 
 & 200 & 74.4 & 0.837 & 1.00 & 0.389 & 0.163 & 0.973 & 0.973 & 0.973 \\ 
 & 500 & 81.7 & 0.837 & 0.95 & 0.389 & 0.163 & 0.976 & 0.975 & 0.976 \\ 
 & 1000 & 83.7 & 0.837 & 0.94 & 0.389 & 0.163 & 0.977 & 0.976 & 0.976 \\  
\bottomrule
\end{tabular}
}
\caption{  
Relationship between $p$ and the convergence rates of iterative methods  for solving Toeplitz systems $\bA \hbb = \bz$ where $\bA =  n  (  \bT +  (n  \sigma_\beta^{2})^{-1} \bI  )$. 
See the main text for the structure of $\bT$. For all the experiments, we used $n  \sigma_\beta^{2} = 250$. 
$\kappa(\bA)$ is the condition number of $\bA$. $\omega^\star_{\rm ICF}$ and $\omega^\star_{\rm SOR}$ refer to the optimal relaxation parameters for ICF and SOR respectively. 
$\rho$ means convergence rate, i.e. the spectral radius of the iteration matrix. An iterative method is convergent  if $\rho < 1$. 
In particular,  $\rho(\bPsi(1))$ refers to ICF with $\omega = 1$; $\rho(\bPsi(\omega^\star))$ refers to ICF with $\omega = \omega^\star_{\rm ICF}$; $\rho_{\rm SOR}$ refers to SOR with $\omega = \omega^\star_{\rm SOR}$; 
$\bar{\rho}_{\rm CG}$ is the upper bound of the convergence rate of CG. 
} 
\label{table:supp.conv1}
\end{table}

\begin{table}[htbp!]
\centering
\noindent
{\small
\begin{tabular}{c lllllllll } 
\toprule
\specialcell{Covariance \\ matrix}  &  $\sigma_\beta$ & $\sqrt{ \kappa(\bA) }$ &  $\omega^\star_{\rm ICF}$  & $\omega^\star_{\rm SOR}$ &  $\rho(\bPsi(1))$  &  $\rho(\bPsi(\omega^\star))$ &  $\rho_{\rm GS}$ &  $\rho_{\rm SOR}$  &  $\bar{\rho}_{\rm CG}$ \\ 
\midrule
\multirow{5}{*}{ \specialcell{$\bT_{c}(0.1)$} }  
 & 10 & 4.82 & 0.999992 & 0.25 & 1.5e-05 & 7.5e-06 & 0.963 & 0.871 & 0.656 \\ 
 & 1 & 4.82 & 0.9992 & 0.25 & 0.0015 & 7.5e-04 & 0.963 & 0.871 & 0.656 \\ 
 & 0.5 & 4.81 & 0.997 & 0.26 & 0.0060 & 0.0030 & 0.963 & 0.871 & 0.656 \\ 
 & 0.2 & 4.76 & 0.982 & 0.26 & 0.038 & 0.018 & 0.961 & 0.868 & 0.653 \\ 
 & 0.1 & 4.58 & 0.931 & 0.27 & 0.148 & 0.069 & 0.956 & 0.858 & 0.642 \\ 
 & 0.05 & 4.05 & 0.788 & 0.33 & 0.539 & 0.212 & 0.932 & 0.821 & 0.604 \\ 
\midrule
\multirow{5}{*}{ \specialcell{$\bT_{e}(5)$} }  
 & 10 & 10.0 & 0.99996 & 1.07 & 8.1e-05 & 4.1e-05 & 0.799 & 0.779 & 0.818 \\ 
 & 1 & 9.96 & 0.996 & 1.08 & 0.0081 & 0.0040 & 0.798 & 0.777 & 0.817 \\ 
 & 0.5 & 9.81 & 0.984 & 1.08 & 0.032 & 0.016 & 0.793 & 0.774 & 0.815 \\ 
 & 0.2 & 8.96 & 0.912 & 1.05 & 0.193 & 0.088 & 0.761 & 0.755 & 0.799 \\ 
 & 0.1 & 7.10 & 0.749 & 0.99 & 0.672 & 0.251 & 0.698 & 0.697 & 0.753 \\ 
 & 0.05 & 4.56 & 0.547 & 0.87 & 1.66 & 0.453 & 0.600 & 0.556 & 0.640 \\ 
\bottomrule
\end{tabular}
}
\caption{  
Relationship between $\sigma_\beta$ and the convergence rates of iterative methods  for solving Toeplitz systems $\bA \hbb = \bz$ where $\bA =  n  (  \bT +  (n  \sigma_\beta^{2})^{-1} \bI  )$. 
 For all the experiments, we used $n   = 1000$ and $p = 200$. 
See Table~\ref{table:supp.conv1} for the annotations of column headers. 
} 
\label{table:supp.conv2}
\end{table}

\subsection{Ill-conditioned systems}\label{sec:ill}
Our last numerical study concerned ill-conditioned linear systems.  Let $\bC$ be the covariance matrix of $\bX$ and consider the linear system $\bA \hbb = \bz$ where 
$$ \bA  = n \bC + \sigma_\beta^{-2} \bI = n \left(  \bC + \dfrac{1}{n \sigma_\beta^2} \bI \right).  $$
The matrix $\bA$ is always full rank due to the shrinkage/regularization term $\sigma_\beta^{-2} \bI$; however, it can be ill-conditioned if $\bC$ is ill-conditioned and $\sigma_\beta$ is large.  

We still fixed $n = 1000, p = 200$ and tried different values for $\sigma_\beta$. 
For $\bC$,  we considered the following three choices so that it is ill-conditioned: Toeplitz matrices $\bT_c(0.95)$ and $\bT_e(500)$ (see the last section for definition) and the matrix $\tilde{\bT}_{c}(20, 0.9)$ defined by 
\begin{align*}
[ \tilde{\bT}_{c}(m, r) ]_{ij} = \left\{ \begin{array}{cc}
1,  &  \quad  \text{ if } i = j , \\
1 , &  \quad \text{ if } i = 2k -1, j = 2k, \, k = 1, \dots, m, \\
1 , &  \quad \text{ if } i = 2k, j = 2k - 1, \, k = 1, \dots, m, \\
r,  & \quad   \text{otherwise} . 
\end{array} \right. 
\end{align*}
That is, $\tilde{\bT}_{c}(20, 0.9)$ represents the scenario where there are $20$ duplicate pairs of covariates ($1$ and $2$, $3$ and $4$, $\dots$, $39$ and $40$), and all the remaining pairs have correlation equal to $0.9$. 
Hence, the matrix $\tilde{\bT}_{c}(20, 0.9)$ has rank equal to $p - 20$. 
The results are summarized in Table~\ref{table:supp.conv3}. 
ICF is still much better than all the other methods except when the covariance matrix is $\bT_e(500)$ and $\sigma_\beta \leq 0.2$. 
For a given covariance matrix, ICF favors larger values of $\sigma_\beta$ while the other methods favor small values of $\sigma_\beta$. 
Hence, we can also conclude that for an ill-conditioned covariance matrix and a fixed $\sigma_\beta$, ICF converges faster for larger sample sizes, which is a very appealing property since the other methods usually fail in such cases. 
Lastly, for the rank-deficient covariance matrix $\tilde{\bT}_{c}(20, 0.9)$, the optimal relaxation parameter for ICF appears to be around $2/3$ and the corresponding convergence rate is about $1/3$. Interestingly, we tried $\tilde{\bT}_{c}(m, r)$ with other values for $m$ and $r$, and made the same observation.

\begin{table}[htbp!]
\centering
\noindent
{\small
\begin{tabular}{c lllllllll } 
\toprule
\specialcell{Covariance \\ matrix}  &  $\sigma_\beta$ & $\sqrt{ \kappa(\bA) }$ &  $\omega^\star_{\rm ICF}$  & $\omega^\star_{\rm SOR}$ &  $\rho(\bPsi(1))$  &  $\rho(\bPsi(\omega^\star))$ &  $\rho_{\rm GS}$ &  $\rho_{\rm SOR}$  &  $\bar{\rho}_{\rm CG}$ \\ 
\midrule
\multirow{5}{*}{ \specialcell{ $\bT_c(0.95)$} }   
 & 10 & 61.6 & 0.9997 & 0.03 & 5.2e-04 & 2.6e-04 & 0.99997 & 0.9992 & 0.968 \\ 
 & 1 & 61.0 & 0.974 & 0.03 & 0.053 & 0.026 & 0.99997 & 0.9992 & 0.968 \\ 
 & 0.5 & 59.3 & 0.900 & 0.03 & 0.221 & 0.100 & 0.99997 & 0.9991 & 0.967 \\ 
 & 0.2 & 50.3 & 0.518 & 0.03 & 1.86 & 0.482 & 0.99996 & 0.999 & 0.961 \\ 
 & 0.1 & 35.6 & 0.132 & 0.04 & 13.2 & 0.868 & 0.99991 & 0.998 & 0.945 \\ 
\midrule
\multirow{5}{*}{ \specialcell{ $\bT_e(500)$} }   
 & 10 & 417.4 & 0.995 & 1.05 & 0.0100 & 0.0050 & 0.998 & 0.998 & 0.995 \\ 
 & 1 & 296.6 & 0.668 & 0.92 & 0.995 & 0.332 & 0.997 & 0.997 & 0.993 \\ 
 & 0.5 & 187.6 & 0.336 & 0.71 & 3.95 & 0.664 & 0.996 & 0.995 & 0.989 \\ 
 & 0.2 & 82.3 & 0.080 & 0.42 & 23.0 & 0.920 & 0.996 & 0.987 & 0.976 \\ 
 & 0.1 & 41.7 & 0.030 & 0.27 & 65.4 & 0.970 & 0.996 & 0.972 & 0.953 \\ 
 \midrule
\multirow{5}{*}{ \specialcell{$\tilde{\bT}_{c}(20, 0.9)$} } & 10 & 4244.1 & 0.666 & 1.36 & 1.00 & 0.334 & 0.99998 & 0.99996 & 0.9995 \\ 
 & 1 & 424.4 & 0.664 & 0.32 & 1.01 & 0.336 & 0.99993 & 0.9996 & 0.995 \\ 
 & 0.5 & 212.2 & 0.660 & 0.17 & 1.03 & 0.340 & 0.99993 & 0.9993 & 0.991 \\ 
 & 0.2 & 84.9 & 0.651 & 0.08 & 1.07 & 0.349 & 0.99991 & 0.998 & 0.977 \\ 
 & 0.1 & 42.5 & 0.635 & 0.04 & 1.15 & 0.365 & 0.9999 & 0.996 & 0.954 \\ 
\bottomrule
\end{tabular}
}
\caption{  
Relationship between $\sigma_\beta$ and the convergence rates of iterative methods  for solving ill-conditioned systems.  
 For all the experiments, we used $n   = 1000$ and $p = 200$. 
See Table~\ref{table:supp.conv1} for the annotations of column headers. 
} 
\label{table:supp.conv3}
\end{table}

\begin{table}[htbp!]
\centering
\noindent
{\small
\begin{tabular}{llllllll } \toprule
\multirow{2}{*}{$n$} & \multirow{2}{*}{$p$} & \multicolumn{3}{c}{IND} & \multicolumn{3}{c}{DEP}  \\ 
\cmidrule(lr){3-5}
\cmidrule(lr){6-8}
&  &  mean($\eta_{min}^2$)  &  max($\eta_{min}^2$)  &  mean($\eta_{max}^2$) &  mean($\eta_{min}^2$)  &  max($\eta_{min}^2$) &  mean($\eta_{max}^2$)  \\ 
\midrule
20 & 10 & 0.0020 & 0.023 & 0.34 & 0.0020 & 0.026 & 0.35 \\ 
20 & 20 & 0.00099 & 0.016 & 0.77 & 0.0010 & 0.021 & 0.77 \\ 
\midrule
50 & 10 & 0.00057 & 0.0087 & 0.11 & 0.00066 & 0.0087 & 0.11 \\ 
50 & 20 & 0.00029 & 0.0063 & 0.3 & 0.00035 & 0.0042 & 0.31 \\ 
50 & 50 & 0.00016 & 0.0021 & 0.91 & 0.00016 & 0.0025 & 0.91 \\ 
\midrule
100 & 10 & 0.00017 & 0.0022 & 0.035 & 0.00017 & 0.0020 & 0.035 \\ 
100 & 20 & 9.3e-05 & 0.0014 & 0.097 & 9.5e-05 & 0.0014 & 0.1 \\ 
100 & 50 & 3.8e-05 & 0.00046 & 0.39 & 4.2e-05 & 0.0005 & 0.4 \\ 
\midrule
200 & 10 & 5.2e-05 & 0.00096 & 0.0099 & 5.5e-05 & 0.00059 & 0.011 \\ 
200 & 20 & 2.3e-05 & 0.00029 & 0.028 & 2.8e-05 & 0.00035 & 0.031 \\ 
200 & 50 & 1.1e-05 & 0.00016 & 0.11 & 1.1e-05 & 0.00018 & 0.12 \\ 
\bottomrule
\end{tabular}
}
\caption{  
The distribution of $\eta_{min}^2$ in the IND and DEP datasets. For each pair $(n, p)$, we sampled $1,000$ data matrices and computed the mean and the maximum of the associated values of $\eta_{min}^2$ with $\sigma_\beta = 0.5$.  
The mean of $\eta_{max}^2$ is listed for comparison. 
} 
\label{table:supp.eta}
\end{table}

\section{Numerical evidence for $\eta^2_{\min} \approx 0$  }
Recall that in Section~\ref{sec:icf.conv}, when we introduced our adaptive strategy for choosing $\omega$ for ICF, we assumed that $\eta_{min}$ is zero. 
Using~\eqref{eq:rad.psi} and~\eqref{eq:optimal.w}, one can show that by omitting sufficiently small $\eta_{min}^2$,
the induced error on $\rho(\bPsi(\omega^\star)) $ is $c \eta_{min}^2$ for some $c \in (0, 1)$. 
Hence, it is fine to neglect $\eta_{min}^2$ as long as it is less than, say $0.01$. 
We only need to verify this assumption for even $p$, since if $p$ is odd, $\eta_{min}$ is always zero. 
It turned out that this assumption holds very generally, once we have a moderately large sample size and $p$ is not too small. 

We used the two GWAS datasets described in the main text, IND and DEP, to examine how fast $\eta_{min}^2$ decreases to zero. 
As will be shown later in the results, $\eta_{\min}^2$ decreases as either $n$ or $p$ increases; therefore, we only considered small values for $n$ and $p$ in this study. 
For each pair $(n, p)$, we randomly sampled $1,000$ data matrices $\bX$ from each dataset, and computed the mean and the maximum of the associated values of $\eta_{min}^2$. 
We still used $\sigma_\beta = 0.5$.
The results shown in Table~\ref{table:supp.eta} indicate that it is very safe to assume $\eta_{min}^2 \approx 0$ when the sample size is large, say greater than $100$.
We did the same experiments with our simulated datasets used in Section~\ref{sec:supp.sim1} and made very similar observations. 
We did observe that $\eta_{min}^2$ tends to be larger when the data has strong collinearity or a heavy-tailed distribution (i.e. the log-normal data). 
But as long as $n \geq 100$ and $p \geq 10$, it is safe to assume that $\eta_{min}^2$  is negligible (less than $0.01$). 
Table~\ref{table:supp.eta} also includes the mean of $\eta_{max}^2$. By direct calculations, one can verify that the induced relative error on $\rho(\bPsi(\omega^\star)) $ is approximately $\eta_{min}^2/\eta_{max}^2$, which is still very small in every case we considered.

\section{Summary of the MCMC algorithm of fastBVSR}
\begin{algorithm}[H]
\caption{fastBVSR}\label{alg:fbvsr}
\begin{algorithmic}
\STATE Initialize $\bgamma^{(0)}$ and compute $\bm{Q}(\bgamma^{(0)} ) = \bX_{\bgamma^{(0)}}^t \bX_{\bgamma^{(0)}}$ and its Cholesky decomposition.  
\vspace{0.1cm}
\FOR{$i=1$ \TO $N_{mcmc}$}
\vspace{0.1cm}
  \STATE Propose $h', \bgamma'$ using $h^{(i)}, \bgamma^{(i)}$. 
\vspace{0.1cm} 
  \STATE Update the Cholesky decomposition of $\bm{Q}({\bgamma'})$ from that of $\bm{Q}(\bgamma^{(i)} ) $ (see examples below).  Compound the update multiple times if necessary.
\vspace{0.1cm}
  \STATE Compute $\sigma_\beta^2 = \sigma_\beta^2(\bgamma', h')$ and draw $\tilde{\bbeta}_{\bgamma'}$ from $\MVN(\bm{0},  \sigma^2_\beta  \bm{I} )$.
\vspace{0.1cm}
  \STATE Draw $\tilde{\bepsilon}$ from $\MVN(\bzero, \bI)$  and set $\tilde{\by} = \bX_{\bgamma'} \tilde{\bbeta}_{\bgamma'} + \tilde{\bepsilon}$.
\vspace{0.1cm} 
  \STATE Calculate the acceptance ratio $\alpha$ according to \eqref{eq:alpha2} in the main text. 
\vspace{0.1cm} 
  \STATE   Set $h^{(i+1)} = h'$ and $\bgamma^{(i+1)} = \bgamma'$  with probability $\min\{1, \alpha\}$ and stay  otherwise. 
\vspace{0.1cm}
  \IF {$(i+1) \mod 1000 = 0$}
 \vspace{0.1cm}
  \STATE Sample $\pi^{(i+1)}$ and $\tau^{(i+1)}$ and perform Rao-Blackwellization. 
\vspace{0.1cm}
  \ENDIF
\vspace{0.1cm}
\ENDFOR 
\end{algorithmic}
\end{algorithm}

\section{Examples for updating the Cholesky decomposition}
Suppose we start with 
\[
\bX^t \bX = 
\begin{pmatrix} 
9 &  3 & 6 \\
3  & 5 & 4 \\
6  & 4 & 21 \\
\end{pmatrix}, \text{ and }  
\bR=
\begin{pmatrix}
3  &  1 & 2  \\
0  &  2 & 1 \\
0  &  0 & 4 \\
\end{pmatrix}
\]
such that $\bX^t \bX = \bR^t \bR.$  
 
\paragraph*{Add a covariate}  
To add one covariate, we attach it to the last column of $\bX$ and denote the new matrix by $\bX'$. 
Suppose
\begin{equation*}
(\bX')^t \bX'  = \left[ \begin{array}{cccc}
9 &  3 & 6 & 3 \\
3  & 5 & 4 & 7 \\
6  & 4 & 21 & 9 \\
3  & 7 & 9  & 20 \\
\end{array}
\right].
\end{equation*}
To compute the new Cholesky decomposition $\bR'$, we solve 
\begin{equation*}
\left[\begin{array}{cccc}
3  &  0 & 0  & 0 \\
1  &  2 & 0  & 0 \\
2  &  1 & 4  & 0 \\
r_{41} & r_{42} & r_{43} & r_{44} \\
\end{array}
\right] \left[\begin{array}{c}
r_{41} \\
r_{42} \\
r_{43} \\
r_{44} \\
\end{array}
\right]  = \left[\begin{array}{c}
3  \\
7 \\
9 \\
20 
\end{array}
\right],
\end{equation*}
which requires only one forward substitution to get 
\begin{equation*}
\bR' = \left[\begin{array}{cccc}
3  &  1 & 2 & 1 \\
0  &  2 & 1 & 3 \\
0  &  0 & 4 & 1 \\
0  & 0  & 0 &  3 
\end{array}
\right]. 
\end{equation*}

\paragraph*{Remove a covariate} 
Consider removing the second covariate of $\bX$ and 
denote the new matrix by $\bX'$ to get
\[
(\bX')^t \bX' = 
\begin{pmatrix}
9 &  6 \\
6 & 21 \\
\end{pmatrix}, 
\text{ and } \tilde{\bR} \define
\begin{pmatrix}
3  &   2  \\
0  &   1 \\
0  &   4 \\
\end{pmatrix}
\]
is obtained by removing the second column from $\bR$. 
Note that $\tilde{\bR}^t \tilde{\bR} = (\bX')^t \bX'$. 
To make $\tilde{r}_{32}$ zero, we introduce the Givens rotation matrix
\begin{equation*}
\bG = \left[ \begin{array}{ccc}
1  &  0  &  0 \\
0  &  1/\sqrt{17}  & 4/\sqrt{17}  \\
0  &  -4/\sqrt{17}  & 1/\sqrt{17}
\end{array}
\right]. 
\end{equation*}
Note $\bG^t \bG = \bI$, and  
 $\tilde{\bR}^t \bG^t \bG \tilde{\bR}  =  \tilde{\bR}^t \tilde{\bR} = (\bX')^t \bX'$.
The new Cholesky decomposition is then given by 
\begin{equation*}
\bR' =  \bG \tilde{\bR} = \left[ \begin{array}{cc}
3 &  2 \\
0 & 4.123 \\
0  &  0 \\
\end{array}
\right]
\end{equation*} 
and the bottom row of zeros can be removed without affecting subsequent calculations.

\section{Proof for the exchange algorithm}
We prove that the posterior $P(\bgamma, h  \mid \by)$ is the stationary distribution of the exchange algorithm by checking the detailed balance condition, i.e. 
\begin{align*}
\dfrac{ K(\btheta' \mid \btheta ) }{ K (\btheta \mid \btheta') } 
\dfrac{ \int \min\left\{  \alpha(\btheta, \btheta', \tilde{\by})  ,  1 \right\}   
L(\tilde{\by} ,  \btheta' ) Z(\btheta' )
d \tilde{\by} }{\int \min\left\{  \alpha(\btheta', \btheta, \tilde{\by})  ,  1\right\}   
L(\tilde{\by} ,  \btheta ) Z(\btheta )
d \tilde{\by} }  = \dfrac{ Z(\btheta') L (\by , \btheta ') P (\btheta') }{Z(\btheta) L (\by , \btheta ) P (\btheta) }
\end{align*}  
where $\btheta = (\bgamma, h)$ and 
\begin{align*}
\alpha(\btheta, \btheta', \tilde{\by} ) = \dfrac{K(\btheta \mid \btheta')}{K(\btheta' \mid \btheta)}
\dfrac{ L(\tilde{\by} \mid \btheta) }{L(\tilde{\by} \mid \btheta')  }  
\dfrac{L(\by \mid \btheta')}{ L (\by \mid \btheta)} \dfrac{P(\btheta')}{P(\btheta)}. 
\end{align*}
Let $\mathcal{Y} = \{ \tilde{y}:   \alpha(\btheta, \btheta', \tilde{\by} ) > 1 \}$. Then, 
\begin{align*}
&  K(\btheta' \mid \btheta ) \int \min\left\{  \alpha(\btheta, \btheta', \tilde{\by})  ,  1 \right\}   
L(\tilde{\by} ,  \btheta' ) Z(\btheta' )
d \tilde{\by}  \\
= \; &  \dfrac{ Z(\btheta') }{ L(\by, \btheta) P(\btheta) } 
\left [ L(\by, \btheta) P(\btheta) K(  \btheta' \mid \btheta ) \int_{\mathcal{Y} } 
 L ( \tilde{\by} , \btheta' ) d\tilde{\by}
 +  L(\by, \btheta') P(\btheta') K(  \btheta \mid \btheta' ) \int_{\mathcal{Y}^c } 
 L ( \tilde{\by} , \btheta ) d\tilde{\by}
\right]  ,\\
&  K(\btheta \mid \btheta' ) \int \min\left\{  \alpha(\btheta', \btheta, \tilde{\by})  ,  1 \right\}   
L(\tilde{\by} ,  \btheta ) Z(\btheta)
d \tilde{\by}  \\
= \; &  \dfrac{ Z(\btheta) }{ L(\by, \btheta') P(\btheta') } 
\left [ L(\by, \btheta) P(\btheta) K(  \btheta' \mid \btheta ) \int_{\mathcal{Y} } 
 L ( \tilde{\by} , \btheta' ) d\tilde{\by}
 +  L(\by, \btheta') P(\btheta') K(  \btheta \mid \btheta' ) \int_{\mathcal{Y}^c } 
 L ( \tilde{\by} , \btheta ) d\tilde{\by}
\right]  . 
\end{align*}
Thus the detailed balance condition holds and the exchange algorithm leaves the posterior invariant. 
For a more general proof, see~\cite{andrieu2009pseudo}.


\end{document}